\documentclass[11pt]{article}
\pdfoutput=1
\usepackage{epsfig, amssymb, graphicx, amsmath} 
\usepackage{multicol}	% pacchetto per scrivere su piu colonne
\usepackage{multirow}
\usepackage[dvipsnames]{xcolor}
\definecolor{my_colour}{rgb}{0,0,0} 	% black
\usepackage{diagbox}

\usepackage{empheq} 	% per sistemi equazioni, autoload amsmath
\usepackage{enumitem}
\usepackage{eurosym}	% simbolo dell'euro
\usepackage{wrapfig}	% immagini con testo intorno
\usepackage{float}		% immagini affiancate
\usepackage{footnote}
\usepackage{bm}         % grassetto
\usepackage{subcaption}
\usepackage{physics}
\usepackage{cases}
\usepackage{tabularx}

\usepackage{units} % for "vulgar" 1/2

\usepackage{algorithm}
\usepackage[noend]{algpseudocode}
\usepackage{hhline}

\usepackage{bbm}		% Per indicatrice
\usepackage{booktabs}	% Per tabelle (toprule, etc.)

\usepackage{setspace}

\usepackage{upgreek}
\DeclareMathAlphabet\mathbfcal{OMS}{cmsy}{b}{n}

\newcommand{\conditional}{{\raisebox{0.3ex}{\scalebox{0.6}{$\!\>\complement$}}}}

\usepackage[authoryear]{natbib}
\bibpunct{(}{)}{;}{a}{,}{,}
\usepackage[a4paper,top=1.25cm,bottom=2.5cm,left=1.75cm,right=1.75cm]{geometry}

\usepackage{amsthm}

\newtheoremstyle{my_style}% Name
  {10pt}% Space above         <-- Adds more vertical space
  {10pt}% Space below         <-- Adds more vertical space
  {\itshape}% Body font    <-- Keep body text normal
  {}% Indent amount
  {\sffamily\bfseries}% Header font: <-- Sans-Serif + Bold for strong contrast
  {.}% Punctuation after header
  {.5em}% Space after header
  {}% Header spec

\theoremstyle{my_style}

\newtheorem*{definition}{Definition}
\newtheorem{theorem}{Theorem}
\newtheorem{proposition}[theorem]{Proposition}
\newtheorem{lemma}[theorem]{Lemma}

\newcommand{\E}{\mathbb{E}}

\newcommand{\be}{\begin{equation}}
\newcommand{\en}{\end{equation}}
\newcommand{\ben}{\begin{equation*}}
\newcommand{\enn}{\end{equation*}}
\newcommand{\bea}{\begin{eqnarray}}
\newcommand{\ena}{\end{eqnarray}}

\newcommand{\Erre}{\mathbb{R}}

\newcommand{\filt}{\mathcal{F}}

\DeclareRobustCommand{\rchi}{{\mathpalette\irchi\relax}}
\newcommand{\irchi}[2]{\raisebox{\depth}{$#1\chi$}}

\newcolumntype{Y}{>{\centering\arraybackslash}X}

\usepackage{makecell}

\makeatletter
\newcommand*\bigcdot{\mathpalette\bigcdot@{.5}}
\newcommand*\bigcdot@[2]{\,\mathbin{\vcenter{\hbox{\scalebox{#2}{$\m@th#1\bullet$}}}}\,}
\makeatother

\usepackage{tikz}
\usepackage{tikzscale}
\usetikzlibrary{positioning}
\usetikzlibrary{fit}
\usetikzlibrary{patterns}
\usetikzlibrary{patterns.meta}
\usetikzlibrary{decorations.pathreplacing,calligraphy}
\usetikzlibrary{shapes.geometric, arrows.meta, positioning, shadows}

\usepackage[most]{tcolorbox}

\tcbuselibrary{listings,skins}

\newtcolorbox{proofbox}[1][]{%
  breakable,                    % allow splitting across pages
  colback=gray!10,               % light gray background
  colframe=gray!50,              % lighter edge color
  arc=3mm,                        % radius of corners
  enhanced,
  before skip=6pt,               % space before box
  after skip=6pt,                % space after box
  boxrule=0.5pt,                 % thickness of the border
  #1
}

\begin{document}

\begin{savenotes}
\title{
\bf{Fast and General Simulation\\ 
of Lévy-driven Ornstein-Uhlenbeck processes\\
for Energy Derivatives}}

\author{
Roberto Baviera$^1$
\& Pietro Manzoni$^{1,2}$
}

\maketitle

\vspace*{0.05truein}
\begin{tabular}{ll}
 &  $^1$ Politecnico di Milano, Department of Mathematics, Italy \\
 &  $^2$ University of Edinburgh, Business School, United Kingdom \\
\end{tabular}
\end{savenotes}

\vspace*{0.20truein}

\vspace*{0.11truein}

\begin{abstract}
\normalsize
\setstretch{1.10}

\noindent
Lévy-driven Ornstein-Uhlenbeck (OU) processes 
represent a versatile class of stochastic processes that have garnered interest in the energy sector for their ability to capture typical features of market dynamics. However, 
in the current state of play,
their Monte Carlo simulation is not straightforward for two main reasons: 
\textit{i)} algorithms are only available for some specific processes within this class, and 
\textit{ii)} they are often computationally expensive. 

\noindent
In this paper, we introduce a new simulation technique designed to address both challenges. 
It relies on the numerical inversion of the characteristic function, offering a general methodology applicable to all Lévy-driven OU processes. 
Moreover, leveraging FFT, the proposed methodology ensures fast and accurate simulations, providing a solid basis for the widespread adoption of these processes in the energy sector. 
Lastly, the algorithm allows
\textcolor{my_colour}{explicit}
control of the numerical error.

\noindent
We apply the proposed technique to the pricing of energy derivatives, comparing the results with the existing benchmarks.
Our findings indicate that the proposed methodology is at least one order of magnitude faster than the existing algorithms, while maintaining an equivalent level of accuracy.

\end{abstract}

\vspace*{0.11truein}

\vspace*{0.11truein}

 Computational Techniques, Simulation Modeling.

\phantom{{\bf JEL Classification}:}\,
G13. Contingent Pricing, Futures Pricing.

\vspace{1.3cm}

\begin{flushleft}
\noindent
\textbf{Address for correspondence:}\\[5pt]
\textbf{Pietro Manzoni}\\
Department of Mathematics\\
Politecnico di Milano\\
Piazza Leonardo da Vinci 32\\
I-20133 Milano, Italy\\[1pt]
\texttt{pietro.manzoni@polimi.it}\\
\texttt{pmanzoni@ed.ac.uk}
\end{flushleft}

\newpage

\setstretch{1.20}

\section{Introduction} \label{sec:introduction}

Energy markets display rich dynamics
characterised by mean reversion, jumps, and heavy tails, posing challenges to traditional modelling approaches. 
Classical stochastic models, such as plain Lévy and Gaussian Ornstein-Uhlenbeck (OU) processes,
often fail to adequately capture 
these non-standard behaviours, motivating the exploration of non-Gaussian alternatives
\cite[see, e.g.,][]{benth2008}.
Lévy-driven OU processes have emerged as a promising class of models capable of reproducing these key features. They are built upon the OU stochastic differential equation (SDE)
--~well-known for its mean-reverting behaviour~--
and incorporate Lévy processes, which allow for the presence of jumps and heavy-tailed distributions, enabling a comprehensive representation of energy market dynamics \cite[see, e.g.,][]{hambly2009}.

The objective of this study is to develop an efficient procedure for simulating the trajectories of Lévy-driven OU processes on a discrete time-grid.
Such simulation
is essential for Monte Carlo (MC) methods,
in particular for pricing 
discretely-monitored derivatives --~a class of instruments which is highly traded in energy markets \citep[see, e.g.,][]{fusai2008}.
MC methods are employed by risk managers to derive information about the distribution of their portfolios
and to gain a view on their overall risk profile.
These portfolios often involve several thousands of discretely-monitored options; thus, it is crucial to rely on simulation schemes that are not only accurate but also fast,
because the models need to run for different parameter settings and market conditions.

In the literature, the simulation of Lévy-driven OU processes is realised via Exact Decomposition (ED) algorithms.
These methods decompose a Lévy-driven OU into simpler subprocesses,
for which simulation schemes exist
{ \color{my_colour}
\citep[see, e.g.,][]{zhang2008exact, zhang2009, zhang2011transition, kawai2011, kawai2012, bianchi2017, qu2021, sabino2020ouvg, sabino2022cgmy, sabino2022exact, sabino2023normal, sabino2021gamma, sabino2022fast}.}
However, the ED schemes are not general:
each one is specifically tailored to a particular Lévy-driven OU process, and as a result, they 
are only available
for a very limited set of processes.

\smallskip
In this paper, we present
a novel simulation approach that offers three key advantages:
\textit{i)} 
it is general, as it is applicable to the whole class of Lévy-driven OU processes;
\textit{ii)}
it is fast, achieving runtimes at least one order of magnitude lower than existing ED schemes; 
\textit{iii)} 
it ensures complete control over the numerical error, 
thereby providing a powerful and reliable simulation tool.

The proposed method is based on a double inversion procedure.
Sample paths are simulated by inverting the Cumulative Distribution Function (CDF) of the process increments,
leveraging the Probability Integral Transform.
However, for Lévy-driven OU, the CDF of the increments is not explicitly available, but is known only through the characteristic function (CF).
This requires an additional inversion 
--~from the CF to the CDF~--
which makes path simulation non-trivial, as already documented in the literature on plain Lévy processes 
\citep[see, e.g.,][]{chen2012,ballotta2014}.

The numerical challenges persist, and are in fact compounded, when moving 
from the plain Lévy to the more general Lévy-driven OU framework, posing new difficulties in terms of both algorithmic stability and estimation of the numerical error.
To overcome these issues effectively, 
in our approach we directly reconstruct the CDF of the increments from their complex-shifted CF.

\smallskip
The main contributions of this paper are threefold. First, we develop a \textit{fast} and \textit{general} simulation algorithm for Lévy-driven OU processes that provides explicit control over the numerical error. 
The stability of our method is achieved via a complex-plane shift
--~a technique enabled by new analytical results that we establish in this work~-- yielding accurate and efficient MC simulation.

Second, we design a specialised algorithm for the simulation of an important subclass of Lévy-driven OU processes, namely those with finite-activity (FA), characterised by a discontinuous CDF. 
This algorithm introduces a Bernoulli-based scheme specifically tailored for simulating these processes.

Third, we conduct a thorough experimental comparison of our method against existing ED algorithms in the context of energy derivatives pricing. 
We demonstrate 
that our method achieves equivalent accuracy while delivering runtime at least one order of magnitude faster across different processes 
and parameter regimes.
To our knowledge, this is the first work to successfully apply complex-integration techniques to the simulation of Lévy-driven OU processes.

\smallskip
The rest of the paper is organised as follows. 
Section \ref{sec:overview} provides an overview of the proposed simulation method, illustrating its core steps via a benchmark example.
{\color{my_colour}Section \ref{sec:theprocesses} introduces Lévy-driven OU processes along with new key analytical properties.}
Section \ref{sec:thealgorithm} describes the simulation methodology, presenting our unified framework for the efficient simulation of Lévy-driven OU processes.
Section \ref{sec:ErrorControl} establishes error bounds and discusses implementation details.
Section \ref{sec:applications} presents applications to energy derivative pricing and provides performance comparison against alternative simulation techniques. 
Finally, Section \ref{sec:conclusion} summarises our findings.

\section{Overview of the method} \label{sec:overview}

A Lévy-driven OU process is a mean-reverting process defined as the solution to the SDE
\begin{equation} \label{eq:ou_sde}
    dX_t = -b X_t \, dt + dL_t \,,
\end{equation}
where $L_t$ is a Lévy process and
$b>0$ 
controls the mean-reversion speed 
\citep[see, e.g.,][]{barndorff2001}.
While their theoretical properties are well-established in the literature \citep[see, e.g.,][]{benth2008}, achieving an accurate numerical implementation of Lévy-driven OU processes remains a major challenge, as their distributions over time are generally not accessible in closed form.

How, then, can these processes be simulated efficiently?
The prevailing approach in the literature is to employ ED algorithms
--~schemes that represent a Lévy-driven OU as a sum of components which can be simulated individually.
{ \color{my_colour}
Table \ref{tab:exactliterature} summarises
the main ED algorithms, which primarily focus on two Lévy drivers:
the Tempered Stable (TS) and the Normal Tempered Stable (NTS).
}
Each referenced paper presents an ED algorithm designed for a specific process, often valid only within restricted parameter domains.
Simulations typically rely on rejection sampling, and are therefore time-consuming and highly dependent on the choice of process parameters \citep[see, e.g.,][]{flury1990}.

\smallskip
This paper 
introduces
a fundamentally different approach that applies to all Lévy-driven OU processes, combining generality with a computational speed comparable to that of the Gaussian OU
--~i.e.\ the process driven by a Brownian Motion (BM).
To illustrate our method concretely, it is helpful to 
first focus on this well-known benchmark case: when $L_t$ is a BM, the solution to the SDE \eqref{eq:ou_sde} over any interval $[r, t]$ of length $\Delta t:= t-r$ reads
\begin{equation} \label{eq:overview_useful}
	X_t = X_r e^{-b\Delta t} + \int_r^{t} e^{-b(t-s)} \dd{L_s} \,,
\end{equation}
where the stochastic integral has a centred Gaussian law with variance
$\upsigma^2(\Delta t):=\frac{1-e^{-2b \Delta t}}{2b}$. 
Hence, simulating each increment of the process reduces to sampling a single Gaussian random variable (r.v.).
This approach is simple and efficient, but it 
requires explicit knowledge of the increment distribution
--~a luxury we lack for non-Gaussian OU processes.

Our simulation method overcomes this limitation by requiring only the CF 
of the stochastic integral in \eqref{eq:overview_useful}, which can be easily derived for all Lévy-driven OU processes, as we discuss in the next section.\footnote{
For example, 
in the considered 
Gaussian OU case, the CF takes the familiar form:
$\phi (u, \Delta t)=e^{-\frac{1}{2} u^2 \upsigma^2(\Delta t)}$.
}
The simulation consists of two main steps.
First, we compute the CDF $P(x)$ of the stochastic integral in \eqref{eq:overview_useful} from its CF $\phi(u, \Delta t)$;
the natural starting point for this operation is the basic formula:
\begin{equation} \label{eq:inversiongil}
	P(x)
	=
	\frac{1}{2}
	-
	\frac{1}{2\pi}
	\int_{-\infty}^{+\infty}
		\frac{e^{-iux} \, 	\phi (u, \Delta t) }{iu}
	\dd{u} \,,
\end{equation}
valid when $\phi(u, \Delta t)$ is sufficiently regular 
\citep[cf.][]{gil1951, wendel1961}.
Second, we draw a uniform r.v.\ $U \sim \mathcal{U}(0,1)$
and obtain $X_t$ by computing $P^{-1}(U)$ and adding the deterministic term 
$X_r e^{-b \Delta t}$.

\medskip
The strength of this approach lies in its broad applicability to Lévy-driven OU processes, provided that three technical considerations are properly addressed:

\begin{enumerate}[noitemsep,topsep=0pt,label=(\roman*)]
\item
For a generic Lévy driver, the basic inversion formula in \eqref{eq:inversiongil} is not well-posed, leading to instability in the computation of the CDF.
A general formula for $P(x)$ is obtained by suitably shifting the integration path in the complex plane; we present the details of this procedure in Section \ref{sec:thealgorithm}.

\item
The CDF $P(x)$ must be evaluated on a fine grid $x_1,\dots, x_N$, which is computationally expensive if done point-by-point.
In our method, we perform this task through the
Fast Fourier Transform \citep[FFT;][]{cooley1965}, which efficiently computes all points simultaneously
and yields numerical estimates of the CDF $\widehat{P}(x_1), \dots, \widehat{P}(x_N)$.

\item
For fast sampling, efficient inversion of the CDF is required.
We achieve this via cubic spline interpolation of the computed values 
$\widehat{P}(x_1), \dots, \widehat{P}(x_N)$, 
which reconstructs the CDF over its domain and 
allows immediate evaluation of $\widehat{P}^{-1}(U)$ for any uniform $U$
\citep[see, e.g.,][]{quarteroni2006}.
\end{enumerate}

\smallskip
\noindent
The Basic Algorithm below summarises the procedure.

\begin{tcolorbox}[colback=gray!5, colframe=black!75, 
                  title=Basic Algorithm\label{alg:simulBasic}
                  for generating Lévy-driven OU processes,
                  fonttitle=\bfseries, sharp corners, boxrule=0.8pt]
\begin{algorithmic}[1]
\Require $X_r$ (Value of the process at time $r$), 
$t$ (Target time),
$\Delta t := t - r$ (Simulation step)
\vspace{0.5ex}
\State Compute the CF $\phi(u, \Delta t)$ of the integral process
	$\int_{r}^{t} e^{-b(t -s)} \dd{L_s}$
\State Retrieve the CDF $\widehat{P}$ via FFT inversion of
$\phi(u, \Delta t)$
\State Interpolate the obtained CDF using splines
\State Draw a uniform r.v.\ $U\sim\mathcal{U}(0,1)$ and compute $\widehat{P}^{-1}(U)$
\State Update $X_{t} = X_r e^{-b \Delta t} + \widehat{P}^{-1}(U)$
\end{algorithmic}
\end{tcolorbox}

\begin{table}[!t]
\centering
\small
\resizebox{\textwidth}{!}{
\begin{tabularx}{1.0\textwidth}{|c||Y|Y|c|} 
	%\toprule
	\hline
		\multirow{2}*{\bf{Process}} &
		\multirow{2}*{\bf{Finite-Activity}} &
		\multicolumn{2}{c|} {\bf{Infinite-Activity}} \\ 
			\cline{3-4}
			& 
			& 
			\bf{Finite-Variation} & 
			\bf{Infinite-Variation}
		\\
		\hhline{|=|=|=|=|}
		\bf{OU-TS} &
			Sabino (2022a) &
			%\Tstrut
			\shortstack{
				Qu et al.\ (2021) \vphantom{ \Big[ } \\[-0.18cm]
				Sabino and Cufaro P.\ (2022)
			} &
			$\times$
		\\ \hline
		\bf{TS-OU} &
			Sabino and Cufaro P.\ (2021) &
			\shortstack{ 	
				\phantom{aaa} \\[-0.12cm]
				Zhang and Zhang (2008, 2009) \\
				Kawai and Masuda (2011) \\
				Qu et al.\ (2021) \\
				Sabino and Cufaro P.\ (2022)
			} &
			$\times$
		\\ \hline
		$\vphantom{\Big(}$ \bf{OU-NTS} $\vphantom{\Big)}$ &  
		\shortstack{ \phantom{aaa}  \\[-0.08cm] $\times$} &
		\shortstack{ \phantom{aaa}  \\[-0.08cm] Sabino (2020, 2023)} &
		\shortstack{ \phantom{aaa}  \\[-0.08cm] Sabino (2023)}
		\\ \hline
		\bf{NTS-OU} &
			\shortstack{ 
				\phantom{aaa} \\[-0.13cm]
				Bianchi et al.\ (2017) \\
				Sabino and Cufaro P.\ (2021)
			} &
			Sabino (2022b) &
			Sabino (2022b)
		\\ \hline
\end{tabularx}
}
\caption{
\small{Existing literature on ED algorithms for Lévy-driven OU simulation.
The rows list the four processes considered in these studies 
(OU-TS, TS-OU, OU-NTS, and NTS-OU),
rigorously defined in Appendix  \ref{app:lcf_of_ts_nts};
the columns classify them according to their jump properties.
Finally, the symbol $\times$ indicates that no ED scheme exists in the literature. 
Our general method allows the simulation of any of the processes in the table.}
}
\label{tab:exactliterature}
\end{table}

\medskip
In this section, we have offered an overview of the paper's content, considering a special example of Lévy-driven OU process: the Gaussian OU. The next section broadens the focus to the general case, wherein any Lévy process $L_t$ can be chosen as the stochastic driver in \eqref{eq:ou_sde}, and discusses relevant properties of the CF of Lévy-driven OU processes.
\textcolor{my_colour}{Readers primarily interested in the simulation methodology may proceed directly to Section \ref{sec:thealgorithm}, as the intervening section provides supporting theoretical results that can be consulted as needed.}

\section{The Processes and their Properties} \label{sec:theprocesses}
Let $(\Omega, \mathcal{F}, (\mathcal{F}_t)_t, \mathbb{P})$ be a filtered probability space, 
where $\mathcal{F}_t$ is the available information at time $t$, 
and let $L_t$ be an $\mathcal{F}_t$\,-adapted pure-jump Lévy process with Lévy density $\nu_L(x)$.
For a specified initial condition $X_0$, a Lévy-driven OU process is defined as the unique strong solution $X_t$ of the SDE \eqref{eq:ou_sde}, given by
\begin{equation} \label{eq:soluz_strong}
	X_t = X_0 e^{-bt} + Z_t \,,
	\quad \text{with} \quad
	Z_t := \int_{0}^{t} e^{-b(t-s)} \dd{L_s} \,,
\end{equation}
where the integral process $Z_t$ is an additive process
\citep[see, e.g.,][]{sato1999}. 

\smallskip
As discussed in \citet{barndorff2001},
Lévy-driven OU processes can be constructed in two distinct ways, giving rise to two broad classes of processes: \textit{OU-Lévy} and \textit{Lévy-OU}. The former are constructed by specifying the Lévy driver $L_t$
--~and then deriving the resulting dynamics of $X_t$\,.
The latter are constructed by specifying the stationary distribution of the solution $X_t$ 
--~from which the corresponding Lévy driver $L_t$ is then determined.

{\color{my_colour}
This paper considers in particular 
four prominent processes:
OU-TS and OU-NTS 
(processes of OU-Lévy type, obtained using TS and NTS as Lévy drivers), 
TS-OU and NTS-OU
(processes of Lévy-OU type, obtained using TS and NTS laws as stationary distributions).
These are the processes listed in Table \ref{tab:exactliterature}, for which exact ED schemes exist;
detailed definitions are provided in Appendix \ref{app:lcf_of_ts_nts}.
}

\smallskip
{\color{my_colour}
In this section, we recall the general expression for the CFs of both classes, specifying their relationships with either the Lévy driver (for OU-Lévy) or the stationary distribution (for Lévy-OU).
Most importantly, we prove two new results on the regularity of these CFs, which are instrumental to our simulation method.}

Lévy-driven OU are naturally handled through their CFs. Throughout the paper, we adopt notation consistent with the reference book by
\citet{cont2003}:
given a process $L_t$,
we denote its CF by $\phi_L(u,t)$,
its log-characteristic function (LCF) by $\Psi_L(u,t)$ and, 
if $L_t$ is Lévy,
its characteristic exponent by
$\psi_L(u)$.
In the latter case,
time-homogeneity implies that $\Psi_L(u,t) = t \, \psi_L(u)$.

\subsection{OU-Lévy Processes} \label{sec:oulevy}

Processes in the OU-Lévy class are constructed by specifying the dynamics of the Lévy driver $L_t$.
Since the increments of $L_t$ are stationary and independent,
the LCF of the integral process $Z_t$ in \eqref{eq:soluz_strong} reads
\begin{equation}\label{eq:oulevy-relationship}
	\Psi_Z(u,t) =
	\int_{0}^{t} \psi_L(ue^{-bs}) \dd{s}
	\,,
\end{equation}
where $\psi_L(\bigcdot)$ is the characteristic exponent of the Lévy driver $L_t$ \citep[cf.][Theorem 1, p.411]{lukacs1969}.
The CF of any OU-Lévy can thus be derived directly from the CF of its driver.

\smallskip
An important concept in the theory of CFs is the \textit{analyticity strip}.
Although the CF $\phi(u) = \E[e^{iuY}]$ of a r.v.\ $Y$ is, in its standard formulation, defined for $u \in \mathbb{R}$, it is generally possible to extend its domain to complex arguments by analytic continuation.
The analyticity strip is defined as the maximal horizontal strip in the complex plane, containing the real axis, where $\phi(u)$ is holomorphic
\citep[see, e.g.,][]{lukacs1972, lee2004}. 
We denote this strip by the open interval $(-p_+, p_-)$ with $p_-, p_+ > 0$, corresponding to the range of imaginary parts $\Im(u)$ for which the CF remains analytic.

\smallskip
In the next proposition, we prove a key property regarding the analyticity strip of OU-Lévy processes.
\begin{proposition} \label{prop:oulevystrip}
	Let $Z_t$ be the integral process of an OU-Lévy driven by $L_t$.
	Then, for every $t > 0$, the analyticity strip of $\phi_Z(u,t)$ 
	coincides with that of $\phi_L(u)$.
\end{proposition}
\begin{proof}
	See Appendix \ref{sec:appendixA}
\end{proof}

\noindent
Hence, in the OU-Lévy framework, the integral process $Z_t$ inherits the domain of analyticity
of its Lévy driver $L_t$, implying that
the CF $\phi_Z(u, t)$ can be analytically continued to the same strip as $\phi_L(u)$, while remaining free of singularities.
This new result 
both establishes the validity of complex-integration techniques for OU-Lévy processes and determines the precise region in which they can be applied.\footnote{For the purposes of this work, we focus on processes whose analyticity strip is non-empty. %As well-known from classical results 
%\citep[see, e.g.,][]{lukacs1972}, 
This requirement coincides with requiring that the process has exponentially decaying tails
\citep[see, e.g.,][]{bavieramassaria2026}, 
%a condition
a standard assumption in the literature and 
satisfied by all processes of practical interest 
\citep[see, e.g.,][]{cont2003}.}

\medskip
We also recall a useful structural relationship from the literature on OU-Lévy processes:
the Lévy density of $Z_t$ can be expressed in terms of that of the driver $L_t$ as
\begin{equation} \label{eq:oulevymeasure}
	\nu_Z (x,t)
	=
	\frac{1}{b x}
	\int_x^{xe^{bt}} \nu_L(y) \dd{y} \,,
\end{equation}
with $\nu_Z(0,t) = 0$ \citep[see, e.g.,][Proposition 15.1, p.482]{cont2003}.
Hence, the cumulants\footnote{
{\color{my_colour}
The cumulants $c_k(Y)$ of a r.v.\ $Y$ with LCF $\Psi_Y(u)$ are defined for \(k=1,2,\dots\) \, as
$c_k(Y) := (-i)^k \left. \frac{d^k}{du^k} \Psi_Y(u) 
\right|_{u=0} \,.$
}
}
$c_k(Z_t)$ of the process $Z_t$ can be obtained from the cumulants of the driver at time $t=1$, denoted $c_k(L_1)$:
\begin{equation} \label{eq:oulevycumulant}
	\phantom{\forall k = 1,2,\dots}
	\hspace{1.5cm}
	c_k(Z_t) =  \frac{1-e^{-kbt}}{kb} \, c_k(L_1)
	\hspace{1.5cm}
	\forall k = 1,2,\dots
\end{equation}

\subsection{Lévy-OU Processes} \label{sec:levyou}

Processes in the Lévy-OU class are specified by the stationary distribution of the process $X_t$ --~denoted in the following as the r.v.\ $X$~-- rather than by directly choosing a Lévy driver $L_t$. 
However, not every stationary distribution is allowed: to ensure the existence of a compatible Lévy driver, $X$ must be self-decomposable \citep[cf.][Theorem 1, p.306]{wolfe1982}.

For Lévy-OU processes, the LCF of $Z_t$ in \eqref{eq:soluz_strong} is related to the characteristic exponent of $X$ by the following relationship \citep[cf.][Theorem 1, p.172]{barndorff2001}:
\begin{equation} \label{eq:levyou-relationship}
	\Psi_{Z} (u, t)
	=
	\psi_{X}(u) - \psi_{X}(ue^{-bt})  \,.
\end{equation}
In the next proposition, we establish a new important property for Lévy-OU processes, which is analogous to that given in Proposition \ref{prop:oulevystrip} for the OU-Lévy class.
\begin{proposition} \label{prop:levyoustrip}
	Let $Z_t$ be the integral process of a Lévy-OU with stationary distribution $X$.
	Then, for every $t > 0$, the analyticity strip of $\phi_Z(u,t)$ 
	coincides with that of $\phi_X(u)$.
\end{proposition}
\begin{proof}
	See Appendix \ref{sec:appendixA}
\end{proof}

\noindent
This result ensures that --~also in the Lévy-OU case~-- $Z_t$ possesses a domain of analyticity, which is crucial for the application of complex-integration in our simulation.
Furthermore, this domain is straightforward to identify, as it coincides precisely with the analyticity strip of the stationary distribution $X$.

\medskip
In addition,
we recall a known relationship between the Lévy density of $Z_t$ and that of the stationary distribution $X$ 
\citep[see, e.g.,][Proposition 1, p.2524]{sabino2022fast}:
\begin{equation}  \label{eq:levyoumeasure}
	\nu_Z(x,t) = \nu_{X} (x) - \nu_{X}(xe^{bt}) e^{bt} \,.
\end{equation}
It follows that the cumulants $c_k(Z_t)$ of the process $Z_t$ can be expressed in terms of those of $X$, as:
\begin{equation}  \label{eq:levyoucumulant}
	\phantom{ \forall k = 1,2,\dots }
	\hspace{1.5cm}
	c_k(Z_t) =  \left( 1-e^{-kbt} \right) \, c_k ( X )
	\hspace{1.5cm}
	\forall k = 1,2,\dots
\end{equation}

\medskip
{\color{my_colour}
In this section, we have reviewed the key properties of Lévy-driven OU processes and proved new results on their analyticity strips. The former are essential for deriving the LCF, Lévy density, and cumulants of any such process (in particular, explicit LCFs for the main processes are available in Appendix \ref{app:lcf_of_ts_nts}); the latter are instrumental for the simulation algorithm presented in the next section.
}

\section{The Simulation Algorithm} \label{sec:thealgorithm}

The simulation of the process $X_t$ in \eqref{eq:soluz_strong}
on the time-grid
$t_0, t_1, \dots, t_Q$
(generally, the reset dates of the discretely-monitored derivative)
leverages the exact solution of the SDE.
Starting from the given initial condition $X_{t_0} = X_0$,
we generate the sample paths by iterating, for $j = 0, \, \dots, \, Q-1$, 
the equation:
\begin{equation}\label{eq:simul_scheme}
	X_{t_{j+1}} = 
	e^{-b \Delta t_j} X_{t_{j}}
	+ Z_{\Delta t_j}	
	\quad \quad 
	\textnormal{with} 
	\quad
	Z_{\Delta t_j} :=
	\int_{t_{j}}^{t_{j+1}}
		e^{-b(t_{j+1}-s)} \dd{L_s}
\end{equation}
where $\Delta t_j := t_{j+1} - t_{j}$ is the length of the $j$-th time interval.
{ \color{my_colour}
At each step, $X_{t_{j+1}}$ is the sum of a mean-reversion term
(known at time $t_{j}$) and an innovation $Z_{\Delta t_j}$\;\!, independent of the filtration $\filt_{t_j}$\;\!.
Notably, by time-homogeneity of $L_t$\;\!, the law of $Z_{\Delta t_j}$ 
depends only on $\Delta t_j$ and not on the absolute time $t_j$.

\smallskip
The main challenge in the simulation is to effectively sample each r.v.\ $Z_{\Delta t}$\;\! (considered hereinafter as the innovation over a generic time interval $\Delta t$), whose distribution is not available in closed form. In our approach, we proceed as follows.
First, we compute the CF of $Z_{\Delta t}$\,\!:
\begin{equation}
\label{eq:ChFun Z}
    \E \left[
        \exp\left\{
        iu Z_{\Delta t}
        \right\}
    \right]
    =
    \phi_Z(u, \Delta t) \,,
\end{equation}
which is readily available for any 
Lévy-driven OU process by using the relationships \eqref{eq:oulevy-relationship} and \eqref{eq:levyou-relationship}.

Then, we reconstruct numerically the CDF of $Z_{\Delta t}$ from the CF $\phi_Z(u, \Delta t)$, as outlined in Section \ref{sec:overview}.
However, the basic inversion formula \eqref{eq:inversiongil} cannot be applied to the CF of any general Lévy-driven OU, since the integrand function has a singularity at the origin.
Following \citet[][Theorem 5.1, p.12]{lee2004}, we introduce a complex-plane shift $a \neq 0$ in the integration contour, and compute the CDF as:
\begin{equation} \label{eq:leeformula}
	P(x)
	=
	R_a
	-
	\frac{e^{-ax}}{2 \pi}
	\int_{- \infty}^{+\infty}
			\frac{e^{-iux} \, \phi_Z \left(u-ia, \Delta t \, \right)}{i(u-ia)}		
		\dd{u} \,\,,
\end{equation}
\begin{equation} \label{eq:leeformulaRA}
	\text{where} \hspace{0.5cm}
	R_a =
	\left\{
	\begin{array}{l l}
	1 \quad\quad & 
		\text{if} \hspace{1.02cm} 0 < a < p_+ \\[0pt]
	0 \quad\quad & 
		\text{if} \hspace{0.30cm} -p_- < a < 0
	\end{array}
	\right.
\end{equation}
and $(-p_+, p_-)$ denotes the analyticity strip of the CF 
$\phi_Z(u,\Delta t)$.

The incorporation of the complex shift $a \neq 0$ constitutes the key novelty of our approach, and distinguishes it from alternative methods based on the Hilbert transform \citep[which corresponds to the case $a=0$; see, e.g.,][for applications to plain Lévy processes]{chen2012, ballotta2014}.
The shift removes the pole at the origin, a well-known source of numerical instability
\citep[as discussed for instance in][]{taufer2009, ballotta2014}. Conversely,
when computed inside the analyticity strip, the shifted integral \eqref{eq:leeformula} does not present any singularity.\footnote{
\textcolor{my_colour}{Crucially, when computed within the analyticity strip, the CDF obtained via the shifted integral coincides with the original CDF, as the complex-plane shift  modifies the integrand for numerical stability without affecting the result.
}
}

\smallskip
Moreover, when the CF $\phi_Z(\bigcdot, \Delta t)$ vanishes at infinity,
the shifted integral can be truncated and efficiently evaluated via FFT-based integration.
Since
%for any fixed $a$, 
the integrand function is conjugate-symmetric along the integration path,
the true CDF $P(x)$ is approximated with the Riemann sum
\begin{equation} \label{eq:Sigma_N_equation}
	\widehat{P}(x) :=
	R_a - 
	\frac{h e^{-a x}}{\pi}
	\Re \left(
	\sum_{n=0}^{N-1}
		\frac{\phi_Z (\left(n+\frac{1}{2} \right) h-ia,\, \Delta t)}
			{i \left(n+\frac{1}{2} \right) h + a }
		e^{-ix \left(n+\frac{1}{2} \right) h}
	\right) \,,
\end{equation}
taken over $N$ equispaced points with step size $h$ 
\citep[see, e.g.,][p.14]{lee2004}. 
The sum is then computed via FFT, which retrieves the approximated CDF over an entire 
$x$-grid of $N$ points. 
%in a single pass.

\medskip
For this methodology to apply to Lévy-driven OU processes, two technical conditions must be satisfied.
First, the shift $a$ must be selected within the analyticity strip of 
$\phi_Z(u, \Delta t)$.
The new results in Propositions \ref{prop:oulevystrip} and \ref{prop:levyoustrip}, 
which provide explicit characterisation of the analyticity strip for any Lévy-driven OU, are therefore essential for this step, as they identify the admissible range of shifts $a$ for these processes.

Second, the CF $\phi_Z(u, \Delta t)$ must vanish as $\abs{u} \rightarrow \infty$ to guarantee convergence of the truncated integral and the stability of the FFT approximation  \eqref{eq:Sigma_N_equation}. 
A sufficient condition is that the distribution of $Z_{\Delta t}$ is absolutely continuous, i.e.\ it admits a Probability Density Function
\cite[as established by the Riemann–Lebesgue lemma; see, e.g.,][p.10]{reed1975}.
The following proposition addresses this requirement. 

\begin{proposition} \label{prop:abs_continuity_of_Z_t}
    For any time interval $\Delta t > 0$, the distribution of $Z_{\Delta t}$  is either absolutely continuous, or it has exactly one probability atom (at zero) and is absolutely continuous elsewhere. \\
    \noindent
    Moreover, the latter case occurs if and only if 
    $Z_{\Delta t}$ has finite-activity, i.e.\ if 
    $\int_{\mathbb{R}} \, \nu_Z(x, \Delta t) \dd{x} < \infty$.
\end{proposition}
\begin{proof}
\textcolor{my_colour}{
    See Appendix \ref{sec:appendixA}.
}
\end{proof}

\noindent
Hence, whenever $Z_{\Delta t}$ has infinite-activity (IA), the FFT-formulation \eqref{eq:Sigma_N_equation} can be applied directly. On the contrary, when $Z_{\Delta t}$ has FA, the presence of a probability atom prevents the CF from decaying at infinity; for this case, we develop a modified simulation approach,
which enables even the CF of these processes to satisfy the vanishing-at-infinity condition.
}

\smallskip
\textcolor{my_colour}{
In the remainder of this section, we first provide all the details for implementing the simulation scheme in the IA case. Then, we present the novel scheme tailored for the FA case.
Finally, we synthesise the two into the Fast and General Monte Carlo (FGMC) algorithm -- a unified methodology for simulating any Lévy-driven OU process.
}

\subsection{The infinite-activity case} \label{ssec:thealgorithmfgmc}

{ \color{my_colour}
Efficient simulation of Lévy-driven OU relies on fast and accurate evaluation of the inversion formula \eqref{eq:leeformula}.
For IA innovations, the simulation over each time interval $(r,t]$ of length $\Delta t$ is performed as follows.
}

\medskip
\noindent
First, one computes the CF $\phi_Z(u, \Delta t)$ of the innovation $Z_{\Delta t}$ in \eqref{eq:simul_scheme}.
The key results are the relationships \eqref{eq:oulevy-relationship} and
\eqref{eq:levyou-relationship}, which allow the determination of the CF 
$\phi_Z(u, t)$ for any Lévy-driven OU process.
Then, the procedure consists of four steps:
\begin{enumerate}[noitemsep,topsep=2pt, start=1]
\item Retrieve the CDF $\widehat{P}$ on an $x$-grid, by computing \eqref{eq:Sigma_N_equation} via FFT, after applying a suitable complex-plane shift $a$ within the analyticity strip of $\phi_Z$.
Guidelines for selecting this shift are discussed in Section \ref{ssec:aruleofthumbforparameterselection}, in the context of error control.

\item Restrict the $x$-grid to the largest consecutive set of points $\{x_n\}$ for which the approximated CDF $\widehat{P}$ is monotonically increasing and lies in $[0,1]$. This step is crucial to ensure that $\widehat{P}$ represents a valid CDF. The set begins with $x_b$ and ends with $x_e$.

\item
\textcolor{my_colour}{
Reconstruct a continuous approximation of the CDF by performing an interpolation of the known values of
$\widehat{P}$ within the (restricted) range of 
the $x$-grid. The interpolation is performed using standard cubic splines with not-a-knot boundary conditions \citep[see, e.g.,][]{quarteroni2006},\footnote{This is often the default implementation of cubic splines (e.g., MATLAB's \texttt{spline} and Python's \texttt{CubicSpline} routines).} which are well-suited to this setting since the CDF of $Z_{\Delta t}$
is continuous and, in most cases 
--~including all IA Lévy-driven OU processes studied in the literature~-- 
at least $C^1$-continuous.
}

\smallskip
\textcolor{my_colour}{
Outside the range $[x_b, x_e]$, extrapolate the CDF using exponential tails,\footnote{
The exponential decay of the CDF follows from the existence of the analyticity strip $(-p_+,p_-)$.
Extensive numerical experiments have shown that, in general, the proposed extrapolation approach is the most effective;
alternative methodologies include, e.g., matching the theoretical asymptotic decay: $\widehat{P}(x_b) e^{p_- (x-x_b)}$ for $x \le x_b$ and $1 - (1-\widehat{P}(x_e)) e^{-p_+ (x-x_e)}$ for $x \ge x_e$.
} by setting
\[
    \widehat{P}(x) =
    \begin{cases}
    \widehat{P}(x_b) \, e^{\lambda_- (x - x_b)} ~~~~
    & \textnormal{for} ~~ x \le x_b \,, 
    \\
    1 - \bigl( 1 - \widehat{P}(x_e) \bigr) \, e^{-\lambda_+ (x - x_e)} ~~~~~
    & \textnormal{for} ~~x \ge x_e \,,
    \end{cases}
\]
where the tail decay rates $\lambda_-$ and $\lambda_+$ are computed
as
\[
\lambda_- := \frac{\ln \widehat{P}(x_{b+1}) - \ln \widehat{P}(x_b)}{x_{b+1} - x_b} \,,
\qquad
\lambda_+ := \frac{\ln \widehat{P}(x_e) - \ln \widehat{P}(x_{e-1})}{x_e - x_{e-1}} \,.
\]
}

\item 
Draw a uniform
r.v.\ $U \sim \mathcal{U}(0,1)$ and generate the innovation as $Z_{\Delta t} = \widehat{P}^{-1}(U)$.
\end{enumerate}

\smallskip\noindent
Finally,
the new value of the process $X_{t}$ is computed by summing the mean-reversion term $e^{-b \Delta t} X_{r}$ and the innovation term $Z_{\Delta t}$, as specified by the update rule \eqref{eq:simul_scheme}.

\medskip\noindent
The proposed scheme
--~outlined in Section \ref{sec:overview} in its basic structure~-- 
is detailed in Algorithm \ref{alg:simulMC1} below. 

\begin{tcolorbox}[colback=gray!5, colframe=black!75, 
                  title=Algorithm~\refstepcounter{algorithm}\label{alg:simulMC1}\thealgorithm.
                  Generation of $Z_{\Delta t}$ for Lévy-driven OU processes (infinite-activity),
                  %absolute continuous distribution case,
                  fonttitle=\bfseries, sharp corners, boxrule=0.8pt]
\begin{algorithmic}[1]
%\Require $\phi_Z(u, \Delta t)$ 
\Require $\phi_Z(u, \Delta t)$ (CF of the innovation), $\Delta t$ (Simulation step)
\vspace{0.5ex} 
\State Retrieve the CDF $\widehat{P}$ by FFT inversion of
$\phi_Z (u, \Delta t)$, using \eqref{eq:Sigma_N_equation} with a suitable complex-shift $a$
\State Restrict the $x$-grid between $x_b$ and $x_e$
\State Interpolate the CDF using cubic splines and apply exponential extrapolation
\State Draw a uniform r.v.\ $U\sim\mathcal{U}(0,1)$ and compute $Z_{\Delta t}$ by CDF inversion
\end{algorithmic}
\end{tcolorbox}

\smallskip
To obtain a complete trajectory over a time-grid $t_0, t_1, \dots, t_Q$\,, it is sufficient to iterate this procedure for every time interval.
Furthermore, when the time-grid is equally spaced, steps 1 to 3 have to be performed only once:
in such a case, the law of the innovation $Z_{\Delta t}$ is the same for every time interval, allowing for further improvement of the algorithmic efficiency.

\subsection{The finite-activity case} \label{subsec:cpsimul}

{ \color{my_colour}
The applicability of Algorithm \ref{alg:simulMC1}
requires that the CF $\phi_Z(\bigcdot, \Delta t)$ vanishes at infinity.
This condition is not satisfied by FA processes,
since their distributions naturally contain a probability atom (cf.\ Proposition \ref{prop:abs_continuity_of_Z_t}).
In this subsection, we introduce a new algorithm for simulating Lévy-driven OU with FA. 
}

By definition, an additive process $Z_t$ has FA if its Lévy density $\nu_Z(x,t)$ integrates to a finite quantity
\begin{equation}
\Lambda(t) := \int_{-\infty}^{+\infty} \nu_Z(x,t) \dd{x} < \infty
\end{equation}
\citep[see, e.g.,][Section 9]{sato1999}. 
In this case, for any fixed time $t$, the distribution of $Z_t$ admits a representation in law as a compound Poisson r.v.\ with time-dependent parameters:
\begin{equation*} \label{eqn:equivalenceinlawFA}
Z_t \,\overset{\text{law}}{\sim} \, \sum_{k=1}^{K_t} J_{t,k} \,,
\end{equation*}
where $K_t$ is a Poisson r.v.\ with parameter $\Lambda(t)$, and
$\{ J_{t,k} \}_k$ are i.i.d.\ absolutely continuous jumps with density $\nu_Z(x,t) / \Lambda(t)$.
Equivalently, the LCF takes the form
\begin{equation} \label{eq:FA_LCF}
\Psi_Z(u,t) = \Lambda(t) (\phi_J(u,t) - 1),
\quad \text{with} \quad
\phi_J(u, t) := \frac{1}{\Lambda(t)}\int_{-\infty}^{+\infty} e^{iux} \nu_Z(x,t) \dd{x} \,.
\end{equation}

\smallskip
Since $K_t$ follows a Poisson distribution, there is always a positive probability that no jumps occur, namely $\mathbb{P}(K_t = 0) = e^{-\Lambda(t)}$. Therefore, for every $t>0$ the distribution of $Z_t$ has a probability atom at zero, which prevents the CF 
$\phi_Z(u,t)$ from vanishing at infinity \citep[see, e.g.,][Section 2.2]{lukacs1970}.

\medskip
The key idea of our simulation procedure for FA Lévy-driven OU is to isolate this probability atom and focus on the simulation of the residual absolutely continuous part.
By separating the zero-jumps case ($k=0$) from the at-least-one-jump case ($k>0$), we decompose the CF of $Z_t$ as:
\begin{equation*}\label{eq:cf_poisson}
	\phi_Z(u, t)
	=
	\E \left[ e^{iuZ_t} \right] = 
	e^{-\Lambda(t)}
	\sum_{k=0}^{\infty} 
		\frac{\Lambda(t)^k}{k!} \phi_J(u,t)^{k}
	=
	e^{-\Lambda(t)} + 
	\frac{e^{\Lambda(t) \, \phi_J(u,t)} - 1}{e^{\Lambda(t)}-1}
	(1-e^{-\Lambda(t)}) \,.
\end{equation*}
%where, in the summation, we have separated the case of
%no jumps ($k=0$) from that of at least one jump ($k>0$).
%
Hence, for any fixed time $t$, the distribution of $Z_t$ can be expressed as a mixture of two components:
\begin{equation} \label{eq:conditional_decomposition_of_Zt}
	Z_t \overset{\tiny{law}}{\sim} 
	\begin{dcases}
		~~Z^{\conditional}_t ~~~ & \text{if } B_t = 1, \\
		~~0 & \text{if } B_t = 0.
	\end{dcases}
\end{equation}
where $B_t$ is a Bernoulli r.v.\ with success probability $1-e^{-\Lambda(t)}$, indicating whether at least one jump occurs,
and $Z^{\conditional}_t$ represents the total sum of the jumps $\{J_{t,k}\}_k$, 
conditional on at least one jump occurring. In particular, the CF of $Z^{\conditional}_t$
--~which we refer to as the \emph{conditional} CF~-- reads
\begin{equation} \label{eqn:fa_decomp}
	\phi_Z^{\conditional}(u,t) 
	:=
	\E \left[ e^{iuZ_t} \, \vert \, B_t = 1 \right] =
	\frac{e^{\Lambda(t) \, \phi_J(u,t)} - 1}{e^{\Lambda(t)}-1} \,.
\end{equation}
Crucially, $\phi_Z^{\conditional}(u,t)$ vanishes at infinity,\footnote{
This follows because $\phi_J(u,t) \rightarrow 0$ as $\abs{u} \rightarrow \infty$ by the Riemann-Lebesgue lemma, since $\nu_Z(x,t)/\Lambda(t)$ is a valid PDF.
} 
ensuring the applicability of the FFT-based inversion.

\medskip
Thus, for any Lévy-driven OU process with FA, the innovation $Z_{\Delta t}$ is simulated in two steps.
First, we sample the Bernoulli r.v.\ $B_{\Delta t}$, determining whether at least one jump occurs.
Then, if $B_{\Delta t}=1$, we simulate the total jump magnitude $Z_{\Delta t}^{\conditional}$ via Algorithm \ref{alg:simulMC1}.
The procedure is summarised in Algorithm \ref{alg:simul2}.
% ensures that the numerical error remains within the bounds prescribed by Proposition \ref{prop:numerical_error1}.

\begin{tcolorbox}[colback=gray!5, colframe=black!75, 
                  title=Algorithm~\refstepcounter{algorithm}\label{alg:simul2}\thealgorithm .
                  Generation of $Z_{\Delta t}$ for Lévy-driven OU processes (finite-activity),
                  fonttitle=\bfseries, sharp corners, boxrule=0.8pt]
\begin{algorithmic}[1]
\Require $\phi_Z(u, \Delta t)$ (CF of the innovation), $\Delta t$ (Simulation step)
\vspace{0.5ex}
\State Draw a Bernoulli r.v.\ $B_{\Delta t}$ with success probability $1 - e^{-\Lambda(\Delta t)}$
\If{$B_{\Delta t} = 1$} 
    \State Compute the \textit{conditional} CF $\phi_Z^{\conditional}(u, \Delta t)$ as in \eqref{eqn:fa_decomp}
    \State Simulate $Z_{\Delta t}$ from the \textit{conditional} CF $\phi_Z^{\conditional}(u, \Delta t)$, using Algorithm \ref{alg:simulMC1}
\Else
    \State No jumps: $Z_{\Delta t}=0$
\EndIf
%\State Update $X_{t} = X_r e^{-b \Delta t} + Z_{\Delta t}$
\end{algorithmic}
\end{tcolorbox}

\medskip
It is worth emphasising that in most cases this method can be substantially faster than the standard approach used for simulating compound Poisson r.v.s
--~commonly appearing in ED schemes for FA processes~-- 
which involves first simulating the number of jumps $K_t$, and then simulating the size of each of the $K_t$ jumps individually.
The computational cost of this second operation can be relevant when many jumps occur.
In contrast, the proposed algorithm requires sampling at most two uniform r.v.s, since all jumps are simulated at once by numerically reconstructing their ``aggregate'' CDF.

\subsection{Fast and General MC for Lévy-driven OU processes}
The FGMC algorithm is a general methodology that unifies the IA and FA schemes, allowing the simulation
of any Lévy-driven OU on a discrete time-grid. 
It is naturally tailored for pricing discretely-monitored derivatives, as it enables the simulation of the process directly at the relevant monitoring dates.
The algorithm is schematically represented in Figure \ref{fig:FGMC}.

\smallskip
\begin{figure}[H]
    \centering
    \resizebox{0.75\linewidth}{!}{
\tikzset{
    startstop/.style={rectangle, rounded corners=4mm, minimum width=2.8cm, minimum height=0.8cm, align=center, draw=blue!80!black, fill=blue!15, thick},
    process/.style={rectangle, rounded corners=2mm, minimum width=2.8cm, minimum height=0.8cm, align=center, draw=green!80!black, fill=green!15, thick},
    decision/.style={diamond, aspect=1.5, minimum width=2.2cm, minimum height=1cm, align=center, draw=blue!80!black, fill=blue!15, thick},
    io/.style={trapezium, trapezium left angle=65, trapezium right angle=115, minimum width=3cm, minimum height=0.8cm, draw=orange!80!black, fill=orange!20, thick},
    algo/.style={rectangle, rounded corners=2mm, minimum width=2.8cm, minimum height=0.8cm, align=center, draw=purple!80!black, fill=purple!20, thick},
    arrow/.style={-Stealth, thick, draw=gray!60},
    label/.style={midway, fill=white, circle, inner sep=1pt}
}

\begin{tikzpicture}[node distance=1.2cm and 0.8cm]

% Main flow nodes
\node (start) [startstop] at (0,0) {\textbf{Start}};

% Draw trapezium manually
\coordinate (input) at (0,-1.5);
\path[draw=orange!80!black, fill=orange!20, thick] 
    (-4.25,-1.9) -- (-4,-1.1) -- (4.25,-1.1) -- (4,-1.9) -- cycle;
\node at (input) {\textbf{Input}: $(r, X_r)$, $\Delta t$, model type, parameters};
\coordinate (input_up) at (0,-1.1);
\coordinate (input_dw) at (0,-1.9);

\node (decision1) at (0, -3) [decision] {FA?};

% Left Branch
\node (drawB) at (-3.25,-3.75) [process] {Draw $B_{\Delta t}$};
\node (decision2) at (-3.25,-5.5) [decision] {$B_{\Delta t} = 1$?};
\node (computePhiC) at (-6.5,-6.5) [process] {Compute $\phi_Z^{\conditional}(u, \Delta t)$};
\node (algo1c) at (-6.5,-8) [process] {Simulate $Z_{\Delta t}$ via $\phi_Z^{\conditional}(u, \Delta t)$,\\ using Algorithm 1};
\node (setZero) at (0,-6.5) [process] {$Z_{\Delta t} = 0$};

% Right Branch
\node (computePhi) at (5,-4.65) [process] {Compute $\phi_Z(u, \Delta t)$};
\node (algo1) at (5,-6.5) [process] {Simulate $Z_{\Delta t}$ via $\phi_Z(u, \Delta t)$, \\ using Algorithm 1};

% Final Part
\node (computeXt) at (0, -9.00) [process] {$X_t = X_r e^{-b\Delta t} + Z_{\Delta t}$};
\node (end) at (0, -10.30) [startstop] {\textbf{End}};

% Arrows
\draw [arrow] (start) -- (input_up);
\draw [arrow] (input_dw) -- (decision1);
\draw [arrow] (decision1) -| (drawB);
\draw [arrow] (decision1) -| (computePhi);
\draw [arrow] (drawB) -- (decision2);
\draw [arrow] (decision2) -| (setZero);
\draw [arrow] (decision2) -| (computePhiC);
\draw [arrow] (computePhiC) -- (algo1c);
\draw [arrow] (setZero) -- (computeXt);
\draw [arrow] (algo1c) |- (computeXt);
\draw [arrow] (computePhi) -- (algo1);
\draw [arrow] (algo1) |- (computeXt);
\draw [arrow] (computeXt) -- (end);

% Text
\node at (-1.8, -2.8) {\textbf{yes}};
\node at ( 1.8, -2.8) {\textbf{no}};

\node at (-5.6, -5.3) {\textbf{yes}};
\node at (-1.0, -5.3) {\textbf{no}};

\end{tikzpicture}
    }
    \caption{\small{Flowchart of the FGMC algorithm. 
    	This FFT-based algorithm can simulate any Lévy-driven OU process, 
    	both with finite- and infinite-activity.}}
    \label{fig:FGMC}
\end{figure}

\textcolor{my_colour}{
The FGMC requires only three parameters of the numerical scheme (hereinafter, \textit{numerical parameters}) associated with the FFT integration in \eqref{eq:Sigma_N_equation}, which are: 
%can be selected according to the desired numerical precision,
%They are:
the number of points $N$ in the $x$-grid 
(conventionally chosen equal to $N=2^M$, with $M \in \mathbb{N}$), 
the integration step $h$ used in the Fourier space, and 
the complex shift $a$.\footnote{
As standard in the FFT literature \citep[see, e.g.,][]{press2007numerical}, we consider a symmetric $x$-grid, with $x_1 = - x_N$.
}
As we show in the next section, these parameters directly control the numerical error of the method.
}

\textcolor{my_colour}{
While users can select the values of these three numerical parameters based on their specific accuracy requirements, Section~\ref{ssec:aruleofthumbforparameterselection} provides guidelines to support this choice.
In particular, we establish selection rules for $h$ and $a$, 
leaving $M := \log_2 N$ as the only free parameter that controls the resolution of the CDF reconstruction.
}

\medskip
\begin{table}[h]
\centering
%\begin{tabularx}{\textwidth}{|c|l|}
\begin{tabular} {|c|l|}
	\toprule
	{\color{my_colour}\textbf{~~Numerical parameter~~}}
    & \textbf{Description}\\ \bottomrule
		%${\cal N}_{sim}$ & Number of sample paths in the MC simulation \\
		$M$ & Integer number equal to $\log_2 N$, with $N$ number of FFT points ~~\\
		$h$ & Step size in the Fourier domain (integration step) \\
		$a$ & Imaginary shift of the integration path in the complex-plane \\
	\bottomrule		
\end{tabular}
\caption{\small{The three numerical parameters for the proposed FGMC algorithm. 
%with brief descriptions.
}}
\label{tab:ParameterList}
\end{table}

%The FGMC algorithm is highly robust and allows complete error control. 
%While the user can choose these parameters to their specific needs, 
%it is useful to provide initial selection criteria. 
%At the end of Section~\ref{sec:ErrorControl}, we illustrate some rules of thumb in a concrete example.
%following a detailed discussion of error control.

\section{Error Control} \label{sec:ErrorControl}

In this section, we provide general bounds for the numerical error introduced by the FGMC algorithm.
We then derive specific error bounds for the main Lévy-driven OU processes, illustrating how our method translates into precise error control.
Finally, we present practical guidelines for parameter selection.

\subsection{Control of the numerical error} \label{ssec:error_control}
A key theoretical property of the FGMC algorithm is that its numerical error can be explicitly controlled for any CF that vanishes at infinity.
This condition holds for all IA Lévy-driven OU, while for FA processes it can be achieved through a suitable transformation, as we have discussed in Section~\ref{subsec:cpsimul}.

\smallskip
For every fixed $t>0$, we distinguish two main cases, depending on the decay of the CF modulus:
\begin{itemize}[noitemsep,topsep=2pt]
	\item[(i)] the decay is \textit{exponential}, i.e.\
	there exist constants $B,\ell,\omega>0$ 
	such that, for
	sufficiently large $\abs{u}$, the following bound holds
	\begin{equation} \label{eq:expexpexp}
		\abs{\phi_Z(u - ia, t)} < B e^{-\ell \abs{u}^{\omega}}
		\hspace{1cm}
		\forall
		a \in 
		(-p_-, p_+) \,\,;
	\end{equation}
	\item[(ii)] the decay is \textit{power-law}, i.e.\
	there exist constants $B,\omega>0$ 
	such that, for
	sufficiently large $\abs{u}$, the following bound holds
	\begin{equation} \label{eq:powpowpow}
		\abs{\phi_Z(u - ia, t)} < B \abs{u}^{-\omega}
		\hspace{1cm}
		\forall
		a \in 
		(-p_-, p_+) \,\,.
	\end{equation}
\end{itemize}

\noindent
As we demonstrate in the following proposition, 
the numerical error of our algorithm is fully determined by the decay of the CF within the analyticity strip.

\begin{proposition} \label{prop:numerical_error1}
	For every shift
	$a \in \left(\, - p_- /2,  \, p_+ /2 \right) \setminus\{0\}$
	and for every $x \in \Erre$,
	the numerical error $\mathcal{E}(x):= \big\vert P(x)-\widehat{P}(x) \big\vert$ of the CDF approximation satisfies 
	the following bounds:
	\begin{itemize}[noitemsep,topsep=3pt]
	\item[(i)]
	if the modulus of the CF has exponential decay,
	\begin{equation} \label{eq:exponential_error}
		\mathcal{E}(x)
		\leq
		\frac{B e^{-a x}}{\pi \omega}
		\frac{e^{- \ell (Nh)^\omega}}{ \ell (Nh)^\omega}
		+
		\frac{1+e^{-2a x} \phi_Z(-2ia, t)}{e^{2\pi\abs{a}/h} - e^{-2\pi\abs{a}/h}}
		%\left[
		%	1
		%	+
		%	e^{-2a x} \phi_Z(-2ia)
		%\right] 
        \,;
	\end{equation}
	\item[(ii)]
	if the modulus of the CF has power-law decay,
	\begin{equation} \label{eq:powerlaw_error}
		\mathcal{E}(x)
		\leq
		\frac{B e^{-a x}}{\pi \omega}
		(Nh)^{-\omega}
		+
        \frac{1+e^{-2a x} \phi_Z(-2ia, t)}{e^{2\pi\abs{a}/h} - e^{-2\pi\abs{a}/h}} 
        \,.
	\end{equation}
	\end{itemize}
\end{proposition}
\begin{proof}
	See Appendix \ref{sec:appendixA}
\end{proof}

\noindent
The approximation error resulting from FFT integration is controlled by two distinct terms: the first (the truncation error) depends on the tail behaviour of the CF; 
the second (the discretisation error) is related to the discretisation of the integral and is the same in both decay cases.
%
%The tail behaviour of the four main Lévy-driven OU processes considered in this study are summarised in Proposition \ref{prop:numericalerrordecaytsnts}.

{\color{my_colour}
\subsection{An example of error control: the OU-NTS process} \label{ssec:errorcontrol}

To illustrate how the simulation error can be explicitly controlled and related to model parameters, we consider the OU-NTS process \citep[see, e.g.,][]{sabino2023normal} as a representative example. 
Thanks to Proposition \ref{prop:numerical_error1}, it suffices to analyse the decay of $|\phi_Z(u,t)| = \exp\{ \, \Re \left( \Psi_Z(u,t) \right) \, \}$.
%as , which can be equivalently expressed as
%\begin{equation}\label{eq:abs_phi_Z_real}
%	\abs{\phi_Z(u,t)} = 
%	\abs{e^{\Psi_Z(u,t)}} =
%	\abs{e^{\Re \left( \Psi_Z(u,t) \right) }
	%	 e^{i \Im \left( \Psi_Z(u,t) \right) }} = 
	%^{\Re \left( \Psi_Z(u,t) \right) } \,.
%\end{e%quation}

\medskip
The OU-NTS process, whose LCF is reported in Appendix \ref{app:lcf_of_ts_nts}, exhibits different jump regimes depending on the sign of the parameter $\alpha$. For $\alpha > 0$, the process has IA (and its CF vanishes at infinity), allowing direct simulation via Algorithm~\ref{alg:simulMC1}. For $\alpha < 0$, the process has FA, and the simulation requires Algorithm~\ref{alg:simul2}.
%which implements the decomposition presented in Section~\ref{subsec:cpsimul}. 
In this latter case, its LCF can be expressed as that of a compound Poisson r.v.\ \eqref{eq:FA_LCF}:
\begin{equation*}
	\Psi_Z(u,t)
	=
	\underbrace{\frac{(1-\alpha)\, t}{\kappa\abs{\alpha}}}_{\displaystyle \Lambda(t)} \,
	\Bigg[
		\underbrace{
		\left( \frac{1-\alpha}{\kappa} \right)^{-\alpha}
		\frac{1}{bt} \,
		\int_{e^{-bt}}^1
			\left(
				\frac{1}{2} \sigma^2 u^2 z^2
				- i \theta u z
				+ \frac{1-\alpha}{\kappa}
			\right)^\alpha
		\frac{\dd{z}}{z}
		}_{\displaystyle \phi_J(u,t)}
		- 1
	\Bigg] \,.
\end{equation*}
This representation makes it straightforward to construct the conditional CF $\phi_Z^{\conditional}(u,t)$ in \eqref{eqn:fa_decomp}.

\smallskip
To precisely determine the numerical error in both regimes, we establish asymptotic bounds on
the CF decay. The following lemma demonstrates that for the OU-NTS the decay is exponential in the IA
case, and power-law in the FA case.
\begin{lemma} \label{prop:ounts_decay}
	For an OU-NTS process with $\alpha \in (0,1)$,
	the CF $\phi_Z$ has exponential decay as 
	$\abs{u} \to \infty$\,\textnormal{:}
	\begin{equation*}  \label{decay:ounts_decay}
		\log \big| \phi_Z(u - ia,t) \big|
		=
		-
        \left(\frac{\sigma^2}{2}\right)^\alpha
		\left( \frac{1-\alpha}{\kappa} \right) ^{1-\alpha}
		\frac{1 - e^{-2\alpha b t}}{2\alpha^2 b}
		\abs{u}^{2\alpha}
		+ o \left( \abs{u}^{2\alpha} \right)
		\hspace{1cm}
		\forall  a \in (-p_-, p_+) \,.
	\end{equation*}
	
	\noindent
	For an OU-NTS process with $\alpha < 0$, 
	the conditional CF $\phi_Z^{\conditional}$ has power-law decay as 
	$\abs{u} \to \infty$\,\textnormal{:}
	\begin{equation*}  \label{decay:ounts_decay_fa}
        \phantom{-}
        \phantom{\log}
		  \big| \phi_Z^{\conditional}(u - ia,t) \big|
		=
        \left(\frac{\sigma^2}{2}\right)^\alpha
		\left( \frac{1-\alpha}{\kappa} \right) ^{1-\alpha}
		\frac{e^{-2 \alpha b t} - 1}{2 \alpha^2 b}
		\abs{u}^{2\alpha}
		+ o \left( \abs{u}^{2\alpha} \right)
		\hspace{1cm}
		\forall  a \in (-p_-, p_+) \,.
	\end{equation*}
	In both cases, $(-p_+, p_-)$ denotes the NTS analyticity strip 
    (see Lemma \ref{prop:nts_strip}).
\end{lemma}
\begin{proof}
	See Appendix \ref{sec:appendixA}
\end{proof}

\noindent
Hence, for OU-NTS processes, both IA and FA regimes deliver effective error control: 
IA processes achieve a rapid exponential decrease of the numerical error, enabling high-precision simulations, while FA processes
--~when handled using the procedure in Section~\ref{subsec:cpsimul}~--
maintain a robust algebraic decrease suitable for all practical applications.
}

\medskip
Furthermore, similar decay patterns hold for the other main Lévy-driven OU processes considered in the literature (namely OU-TS, TS-OU, NTS-OU), with detailed results provided in Appendices~\ref{sec:appendix_TS} and~\ref{sec:appendix_NTS}. 
Collectively, the behaviour of the numerical error is summarised in the following proposition. 
\begin{proposition} \label{prop:numericalerrordecaytsnts}
For all Lévy-driven OU processes of TS and NTS type,
% considered in this work, 
the numerical error $\mathcal{E}(x)$ produced by the FGMC algorithm exhibits the following behaviour as a function of the FFT grid size:
\begin{itemize}[noitemsep,topsep=2pt]
\item[(i)] for all IA specifications, $\mathcal{E}(x)$ decreases exponentially;
\item[(ii)] for all FA specifications, $\mathcal{E}(x)$ decreases at a power-law rate.
\end{itemize}
\end{proposition}
\begin{proof}
	See Appendix \ref{sec:appendixA}.
\end{proof}

\subsection{A rule of thumb for numerical parameter selection} \label{ssec:aruleofthumbforparameterselection}

\textcolor{my_colour}{
In Table \ref{tab:ParameterList}, we have listed the three numerical parameters required by the FGMC algorithm; in what follows, we provide practical criteria for their selection based on the error analysis presented above.
}

%The computational time is primarily determined by the number of sample paths, $\mathcal{N}_{sim}$, making it the most relevant hyperparameter. 
%In most markets, the typical required accuracy for derivative pricing is on the order of a few basis points.
%Since the numerical error of MC scales as $\mathcal{N}_{sim}^{\;-1/2}$, selecting $\mathcal{N}_{sim} \approx 10^7$ is a conservative and robust choice for most applications.

%\smallskip
%We now discuss the remaining three hyperparameters individually, 
%which play a key role in controlling the accuracy of the FFT-based CDF reconstruction.
Regarding the grid size $N$, as we show in Section \ref{sec:applications}, 
even a relatively small grid ($M=16$) achieves excellent accuracy.
This indicates that the method can deliver high-precision results without requiring an excessively refined FFT grid. 
%highlighting its efficiency and practical appeal.

Proposition \ref{prop:numerical_error1} provides useful guidance for selecting the shift $a$ and the integration step $h$ for a given grid size $N$, thereby determining all parameters in the numerical scheme, as well as the convergence rate of the error bound.
Specifically, the selection of $h$ is informed by the decay of the CF:
\begin{itemize}[noitemsep,topsep=2pt]
\item[(i)]
In the exponential decay case, for fixed $N$, the integration step $h$ can be chosen as
\begin{equation*}
	h(N) = \max \left(
	\left(
		\frac{2\pi \abs{a}}{\ell N^{\omega}}
	\right)^{{1}/{(\omega+1)}}, \,\, 1 \% \right) \,.
\end{equation*}
The first term ensures that the exponential decay of the truncation error, i.e.\ $e^{-\ell (Nh)^\omega}$, is matched with that of the discretisation error, i.e.\ $e^{-2\pi |a|/h}$.
With this choice of $h$, the overall error bound in \eqref{eq:exponential_error} decreases exponentially with the FFT grid size $N$.\,\footnote{
The 1\% lower bound is a conservative safeguard; 
as we have observed in extensive numerical experiments, even smaller values of $h$ consistently yield excellent accuracy. Further details are available upon request.
}
%In any case, we have observed in extensive numerical experiments, that a selection for $h$ lower than $1 \%$ guarantees excellent results.
%With this choice, the error bound decreases exponentially with the FFT grid size $N$. %, as $O(e^{-\ell N^{\omega/(1+\omega)}})$.

\item[(ii)]
In the power-law decay case, the truncation error 
--~the first term in \eqref{eq:powerlaw_error}~-- dominates, because it decays more slowly than the discretisation error. Based on our numerical experiments, a choice of 
$h = 1 \%$ proves sufficient to ensure accurate results across all tested cases.
With this selection of the parameters, the error bound decreases with $N$ as 
$O(N^{-\omega})$.
\end{itemize}

\smallskip
To minimise the total error on either tail of the CDF, the shift parameter $a$
should be chosen as large as possible within the range permitted by Proposition \ref{prop:numerical_error1}, i.e.\ near either $-p_-/2$ or $p_+/2$.
This places the integration path approximately midway between the singularity at zero and the corresponding boundary of the analyticity strip.

Finally, we suggest a practical rule of thumb to validate the implementation: verify that the condition $\max \big( \,\widehat{P} (x_b),\,\, 1- \widehat{P} (x_e) \,\big) \le 10^{-4}$ is satisfied. 
This ensures the computational domain $(x_b, x_e)$ is sufficiently wide, as it confirms that negligible probability lies outside the boundaries. 
%When this condition holds, the FGMC method returns highly accurate results.

\textcolor{my_colour}{
If this condition is not met, accuracy can be improved by performing two separate FFT reconstructions with different shifts
$a$ -- one positive and one negative -- for the respective positive and negative halves of the CDF grid. 
This trick can enhance the CDF reconstruction while keeping the total MC computational time substantially unchanged.
}

\section{Applications and Numerical Results}\label{sec:applications}
In this section, we analyse the numerical performances of the proposed simulation algorithm and we present two applications in energy derivative pricing.\footnote{All results are obtained using MATLAB R2025a on a machine with
Apple M1 Chip, 8 GB RAM, macOS Sequoia 15.6.1.}

First, we test the accuracy and speed of our simulation method, 
comparing the results with those obtained using the ED algorithms listed in Table \ref{tab:exactliterature}.
The simulation accuracy is measured, as standard in the literature, by considering the first four cumulants: we compare the MC estimates (obtained with $10^7$ simulations) with the true cumulants 
--~ which can be computed analytically for all considered models (see Section \ref{sec:theprocesses}).
The speed is measured as the time in seconds required for the simulation.
%we compare the results of our FGMC algorithm with those of ED.
%Not only do we show that FGMC is much faster than the ED algorithms, but we also observe that FGMC simulation times have the same order of magnitude of the benchmark model, the Gaussian OU, 
%for which an ultra-fast simulation method exists 
%\citep[the Ziggurat method of][]{marsaglia2000}. 

Then, we illustrate how the proposed technique can be effectively used for pricing energy derivatives, considering two common instruments in the energy markets: European and Asian options.

When dealing with options, it is not only important to assess numerical performance in pricing at-the-money (ATM) instruments, but also essential to verify that models and algorithms can correctly reproduce prices across different levels of moneyness \citep[see, e.g.,][]{azzone2023}.
Indeed, the inherent seasonality of energy markets introduces fluctuations in the moneyness of the options over months and years \citep[see, e.g.,][]{sabino2023normal}.
While the existing literature on simulation techniques for these processes focuses on ATM options
\citep[see, e.g.,][]{benth2018lavagnini, sabino2022exact}, 
we perform an extended analysis by considering options in a range of moneyness around the ATM, covering the spectrum of strikes relevant for typical market volatility.
%We examine a panel of 30 options with moneyness in the range 
%$\sqrt{T}(-0.2,0.2)$, where $T$ is the expiry of the option, 
%However, as already observed in the literature, deep out-of-the-money (OTM) and deep in-the-money (ITM) options are not very informative on the method performances, because the option value is close to the intrinsic value
%\citep[see, e.g.,][]{azzone2023}. 

\subsection{Accuracy and Computational Time} \label{ssec:accuracy}

We consider for the comparison all processes listed in Table \ref{tab:exactliterature}. When more than one ED scheme exists, we implement the fastest (and most recently published) algorithm for each class:
\textit{Algorithm 1} and \textit{Algorithm 2} in \citet{sabino2022fast} for 
finite-variation OU-TS and TS-OU,
\textit{Algorithm 1} in \citet{sabino2022cgmy} for finite-activity OU-TS,
\textit{Algorithm 3.1} in \citet{sabino2023normal} for 
finite- and infinite-variation OU-NTS, and
\textit{Algorithm 1} in \citet{sabino2022exact} for 
finite- and infinite-variation NTS-OU.\footnote{
To sample the required TS r.v.s, besides the standard acceptance/rejection algorithm in \citet{brix1999}, 
we have implemented the fast algorithms proposed in \citet{devroye2009} and \citet{hofert2011}, following the joint implementations suggested in the latter paper.
These algorithms for simulating TS r.v.s can be up to 50 times faster than the one in \citet{brix1999} for the parameters considered in the numerical analysis, but in any case they are at least one order of magnitude slower than the proposed FGMC algorithm. Details are available upon request.
}

All processes are simulated on a yearly time-horizon ($t=1$), with the initial condition $X_0=0$. % as in \citet{sabino2022fast}.
For OU-TS and TS-OU processes, as in 
%\citet{poirot2006} and 
\citet{sabino2022fast}, we select the parameters
$(b, \beta_p, \beta_n, c_p, c_n, \gamma_c) = 
(0.1, 2.5, 3.5, 0.5, 1, 0)$ and
let $\alpha_p$ and $\alpha_n$ vary over the range of admissible values, setting
$\alpha := \alpha_p=\alpha_n$.
For OU-NTS and NTS-OU processes, following \citet{cummins2017}, we consider the parameters $(b, \sigma, \kappa, \theta) = (0.2162, 0.201, 0.256, 0)$
and let $\alpha$ vary.\footnote{Similar conclusions have been derived with different parameter sets and are available upon request.}
Finally, for the FFT numerical parameters, 
we choose $M = 16$, while $h$ and $a$ are selected so as to minimise the numerical error, as suggested in Section \ref{ssec:aruleofthumbforparameterselection}.

\smallskip
Table \ref{tab:cumulants_plain} shows the accuracy results, obtained by simulating $10^7$ realisations for each process and for different choices of $\alpha$.
The true cumulants are deduced using the relationships \eqref{eq:oulevycumulant} and \eqref{eq:levyoucumulant}, together with the formulas in Lemmas \ref{prop:TScumulants} and \ref{prop:NTScumulants}.
In all cases, our FGMC simulation algorithm appears to be extremely accurate, reproducing the true cumulants up to the fourth decimal digit.

Table \ref{tab:computational_times_plain} highlights the other important result: the improvement in terms of computational time is extremely relevant.
Across all processes, our simulation technique consistently outperforms the ED schemes, being at least one order of magnitude faster, and the computational times remain relatively constant across the values of the parameter $\alpha$.
This is not true for the ED methods, where the results are highly sensitive to parameter choice, as a consequence of the acceptance/rejection algorithms.

%In addition, we provide the time required 
For reference, the time required
to simulate $10^7$ samples of the Gaussian OU --~the simplest OU process~-- is 0.17 seconds. Hence, FGMC achieves results that are comparable in magnitude to this benchmark model, which is simulated by sampling only a Gaussian r.v.

\smallskip
Finally, the FGMC allows the extension of the MC simulation also to the asymmetric Lévy-driven OU processes of NTS type.
The existing literature on ED simulation has focused exclusively on the symmetric case \citep[cf.][]{sabino2022exact} with the only exception of $\alpha=0$ \citep{sabino2020ouvg}.

In Table \ref{tab:cumulants_asymmetric}, 
we report the results for the simulation of asymmetric OU-NTS and NTS-OU.
We consider a yearly horizon and model parameters equal to
$(b, \sigma, \kappa, \theta) = (0.2162, 0.201, 0.256, 0.1)$, maintaining consistency with those used previously, with the only difference being the parameter $\theta$.
Also in this case, the empirical estimates of the cumulants demonstrate exceptional accuracy, with precision extended up to the fourth decimal digit consistently across all levels of the stability parameter $\alpha$.
The availability of a simulation scheme for such processes opens up a broader variety of modelling possibilities, which have not yet been explored in the current state of play.

\subsection{Energy derivative pricing}

In this subsection, we discuss the application of FGMC in energy derivative pricing.
We assume that the market model is specified in the risk-neutral measure and that the spot price is driven by the following one-factor process:
\begin{equation} \label{eq:onefactormodel}
	S_t 
	=
	F_{0,t} \,
	e^{
		X_t \,
		+ f_t
	} 
	\,,
\end{equation}
where $X_t$ is a Lévy-driven OU starting in $X_0=0$, 
$f_t$ is a deterministic function of time, and $F_{0,t}$ is the forward price with maturity $t$ derived from quoted products at value date $t_0=0$. % which reflects the seasonality.
As carried out by \citet{hambly2009}, the risk-neutral conditions are met by imposing that $f_t = - \Psi_Z(-i, t)$ for all $t \geq 0$.
%\begin{equation*}
%	f_t = - \Psi_Z(-i, t) \,.
%\end{equation*}

\smallskip
We consider in the following the pricing of two main derivatives, a European call option and an Asian path-dependent option, demonstrating
that in both cases FGMC is an extremely well-suited methodology for pricing in terms of both accuracy and computational time.
%
%
%The aim of this analysis is to set the groundwork for pricing energy derivatives and to provide a simple and accurate simulation method, though these models can find applications in many different financial contexts \citep[see, e.g.,][]{cont2003}.
%
Although our analysis focuses on one-factor models, the simulation methodology extends straightforwardly to the multi-factor case, as well as to the modelling of the forward price 
\citep[see, e.g.,][]{benth2019mean, latini2019mean}, %piccirilli2021capturing
as already discussed in the existing literature \citep[see, e.g.,][]{sabino2023normal}.

\subsubsection{European Option}

In the European option case, we evaluate calls with different moneyness levels for the four processes described in the previous sections.
To ensure our simulation technique correctly prices both in-the-money (ITM) and out-of-the-money (OTM) options, we select 30 call options (with one-year maturity, i.e.\ $T=1$), having moneyness $\rchi$ in the range $\sqrt{T}(-0.2,0.2)$.
We recall that the moneyness of an option having strike $K$ and maturity $T$ is defined as
\[
	\rchi := \log \frac{K}{F_{0,T}} \,.
\]
Moreover, as in \citet{sabino2022exact}, we consider a flat forward curve, assuming $F_{0,t}=1 ~ \forall t\geq 0$. 

To ensure consistency with the previous experiments, we opt for the same set of model parameters used in Section \ref{ssec:accuracy}. In particular, we consider symmetric OU-NTS and NTS-OU processes
(i.e.\ $\theta=0$) to allow a direct comparison with ED algorithms.
For the FGMC, the selected numerical parameters are 
$M = 16$, and both $h$ and $a$ as suggested in Section \ref{ssec:aruleofthumbforparameterselection}. The number of simulations is set to $10^7$.

\smallskip
For European call options an exact formula exists
\citep[cf.][]{lewis2001simple}, allowing us to verify whether the retrieved MC prices align with the true prices of the contracts.
To properly assess the accuracy of FGMC, we compute different metrics
--~the Maximum absolute error (MAX), the Root Mean Squared Error (RMSE), the Mean Absolute Percentage Error (MAPE) and the average of the MC Standard Deviation 
$\left(\overline{\textnormal{SD}}\right)$~--
over the set of 30 moneyness levels.

\smallskip
%Table \ref{tab:eu_call} shows the results of the comparison in terms of accuracy, together with the computational time required for the simulation; as in the previous experiments, $10^7$ trials are generated.
Table \ref{tab:eu_call} shows the comparison in terms of accuracy and computational time.
The results are noteworthy: we observe that FGMC achieves an accuracy on the order of one basis point (bp), not only for a few specific moneyness levels, but for all considered levels.
Moreover, this high accuracy is consistent across all four processes and for all choices of the parameter $\alpha$.

This analysis highlights once again the two main strengths of our methodology. First of all, it is fast: the computational time is always remarkably faster than the corresponding ED algorithm.
Second, it is general, and we are indeed able to compute the prices for all the admissible choices of the parameters, whereas --~in many cases~-- the ED algorithms do not allow the computation of the price.

\subsubsection{Asian Option}

As a second application in derivative pricing,
we consider an Asian call option, as in \citet{sabino2022cgmy}. 
In gas and power markets, these contracts typically have expiry dates spanning several years, and 
%are commonly offered by brokers.
%Options of this kind %, %belonging to the category of discretely-monitored derivative, 
play an important role as effective 
hedging devices 
%in commodity markets 
\citep[see, e.g.,][]{fusai2008}.
The payoff of the derivative is
\begin{equation}
	%A_{0,T}(K) = 
    \max\left( \frac{1}{Q} \sum_{j=1}^{Q} S_{t_j} - K, ~0 \right)\,
\end{equation}
where $S_t$ is the spot price, specified as in \eqref{eq:onefactormodel}, $K$ the strike price,
$\{t_j\}_{j=1,\dots,Q}$ the $Q$ monitoring dates, 
and $T=t_Q$ the expiry of the option.
We consider in our analysis an Asian option having two-year expiry and monthly-monitoring, i.e.\ $Q = 24$.

Unlike the European call case, for these derivatives no exact pricing formula exists, and MC simulation is required.
To evaluate such a contract, one needs to simulate discrete trajectories of the process $X_{t_j}$, with
$0=t_0 < t_1 < \dots < t_{24}=T$, where $T=2$ represents the contract expiry.

\smallskip
We conduct a comprehensive comparison of pricing accuracy using FGMC and ED methods. To evaluate performance across different moneyness levels, we focus on three key contracts:
an OTM option ($\rchi=0.2 \sqrt{T}$), 
an ATM option ($\rchi=0$), 
and an ITM option ($\rchi=-0.2 \sqrt{T}$).
As before, we set a flat forward curve
$F_{0,t} = 1 ~ \forall t \geq 0$, and we select model parameters and FGMC numerical parameters as in the previous experiments.

Table \ref{tab:asian_call} presents the results of our analysis. For both simulation methods, we provide the prices of the three contracts along with their corresponding MC standard deviations, and the simulation time.
We observe that, for all considered cases, the obtained option prices align between the two simulation approaches.
Furthermore, the SD consistently remains around $1$ bp, the typical accuracy required in most financial markets, %This fact shows the robustness of the MC estimates of prices, 
confirming the suitability of the technique for real-world applications.

In terms of computational time, the results are remarkable and mirror those observed previously. 
%in the previous experiments. 
As expected, when simulating trajectories with multiple monitoring dates, FGMC showcases a substantial gain in computational speed compared to ED methods, while maintaining analogous pricing accuracy.

\section{Conclusion}
\label{sec:conclusion}

In this paper, we have introduced a new unified framework for simulating Lévy-driven OU, a class of processes which is extremely relevant in the energy sector due to its ability to capture key features of market dynamics. 
The proposed approach, named FGMC, relies on the reconstruction of the CDF from the CF of the stochastic process, leveraging 
complex-integration, FFT, and spline interpolation.
Following a detailed discussion of the method, we have performed an empirical analysis of its accuracy and computational efficiency, and
%FGMC is tailored for the evaluation of any discretely-monitored derivative, as it allows the simulation of the increments between dates in a generic time-grid.
we have illustrated two applications in option pricing. %focusing on two common instruments in energy markets.
%we have compared FGMC performance against ED algorithms, which represent the state-of-the-art techniques for simulating Lévy-driven OU processes.

\smallskip
There are three main contributions of this paper.
First, we have developed a \textit{fast} and \textit{general} simulation method for Lévy-driven OU processes. %that provides rigorous control over numerical error.
Algorithmic stability is ensured through complex-plane shifts, leveraging the novel results on CF regularity for these processes established in Propositions 
\ref{prop:oulevystrip} and \ref{prop:levyoustrip}.
Notably, we have also demonstrated FGMC applicability to novel subclasses of processes,
%--~infinite-activity OU-TS and TS-OU, finite-activity OU-NTS, asymmetric OU-NTS and NTS-OU~--
where no alternative simulation algorithms currently exist.

Second, we have designed a specialised algorithm for FA Lévy-driven OU processes that incorporates a Bernoulli-based correction to optimally handle discontinuous CDFs.
For both simulation algorithms, we have demonstrated that the CDF numerical error $\mathcal{E}(x)$ can be explicitly computed and controlled.
Moreover, we have shown that for all main processes considered in the literature, the error 
$\mathcal{E}(x)$ decays exponentially with the FFT grid size for infinite-activity processes, and exhibits power-law decay for finite-activity processes, as established in Proposition \ref{prop:numericalerrordecaytsnts}. 

Third, we have conducted a comprehensive experimental comparison of our method against the existing ED algorithms, which represent the state-of-the-art techniques for simulating Lévy-driven OU processes.
The results show that, in all cases, FGMC achieves equivalent accuracy while delivering runtimes at least one order of magnitude faster.
Simulation times are consistently comparable to those of the Gaussian OU
--~the least computationally demanding process within the entire OU family~--
rendering FGMC suitable for real-time applications.
Furthermore, the proposed algorithms can be easily parallelised across multiple cores, offering additional opportunities for performance gains.

\textcolor{my_colour}{
A natural extension of this work concerns high-frequency sampling with very small time increments. %$\Delta t$. 
In such settings, the FFT inversion may become numerically sensitive, as the distribution of the process increments increasingly concentrates around zero. The development of robust FFT-based simulation schemes suitable for this regime constitutes a promising direction for future research.
}

\section*{Acknowledgments}
\noindent 
We are grateful to M.\ Azzone,  S.\ Castellano, A.\ Pallavicini, L.\ Torricelli, and F.\ Ziel for their valuable suggestions. 
We also wish to thank all participants of the XXV Quantitative Finance Workshop (QFW) in Bologna,
of the Energy Finance Italia Conference (EFI9) in Bari, and of the International Ruhr Energy Conference (INREC13) in Essen for their insightful feedback. 
Finally, a special thank goes to the R33 team at
EDF R\&D Paris-Saclay, where part of this research has been developed.
%This research is part of the activities of the ``Dipartimento di Eccellenza 2023-2027''. 
%The authors are members of the INdAM Research group GNCS.

\section*{Disclosure of interest}
The authors have no relevant financial or non-financial interests to disclose.

\begin{table}[p]
\centering
\textbf{OU-TS} \\[0.1cm]
\resizebox{\textwidth}{!}{
\begin{tabular}{|c|rrr|rrr|rrr|rrr|}
	\toprule
    \multirow{2}{*}{$\alpha$} & 
    \multicolumn{3}{c|}
    	{$c_1(X_t) \cdot 10^3$} & 
    \multicolumn{3}{c|}
    	{$c_2(X_t) \cdot 10^3$} & 
    \multicolumn{3}{c|}
    	{$c_3(X_t) \cdot 10^3$} & 
    \multicolumn{3}{c|}
    	{$c_4(X_t) \cdot 10^3$}
    \\[0.1cm]
	& 
	True & FGMC & ED~ &
	True & FGMC & ED~ &
	True & FGMC & ED~ &
	True & FGMC & ED~ \\
	\specialrule{\cmidrulewidth}{0pt}{0pt}
    1.6 & 
		0.000 & -0.627 & \multicolumn{1}{c|}{$\times$} &
    	1914.8 & 1914.3 & \multicolumn{1}{c|}{$\times$} &
    	-26.426 & -27.566 & \multicolumn{1}{c|}{$\times$} &
    	107.41 & 107.58 & \multicolumn{1}{c|}{$\times$}\\
    1.2 & 
    	0.000 & -0.365 & \multicolumn{1}{c|}{$\times$} & 
    	640.81 & 640.65 & \multicolumn{1}{c|}{$\times$} &
    	-7.070 & -7.309 & \multicolumn{1}{c|}{$\times$} &
    	94.514 & 94.547 & \multicolumn{1}{c|}{$\times$} \\
    0.8 & 
    	0.000   & -0.264  & -0.027  &
    	323.64  & 323.54  &  323.51 &
    	2.916   &  2.799  &  2.660  &
    	89.495  & 89.532  &  89.836 \\
    0.4 & 
    	0.000 & -0.213 & 0.033 & 
    	202.58 & 202.51 & 202.70 & 
    	9.473 & 9.381 & 9.202 & 
    	90.270 & 90.327 & 90.784 \\
    -1.0 & 
    	0.000 & -0.204 & -0.028 & 
    	100.28 & 100.32 & 100.18 & 
    	31.807 & 31.743 & 31.642 & 
    	138.94 & 139.50 & 138.45 \\
    -2.0 & 
    	0.000 & -0.061 & -0.100 & 
    	105.85 & 105.76 & 105.58 & 
    	66.683 & 66.533 & 66.179 & 
    	256.36 & 256.98 & 254.28 \\
    \bottomrule
\end{tabular}
}
\\[0.5cm]

\textbf{TS-OU} \\[0.1cm]
\resizebox{\textwidth}{!}{
\begin{tabular}{|c|rrr|rrr|rrr|rrr|}
	\toprule
    \multirow{2}{*}{$\alpha$} & 
    \multicolumn{3}{c|}
    	{$c_1(X_t) \cdot 10^3$} & 
    \multicolumn{3}{c|}
    	{$c_2(X_t) \cdot 10^3$} & 
    \multicolumn{3}{c|}
    	{$c_3(X_t) \cdot 10^3$} & 
    \multicolumn{3}{c|}
    	{$c_4(X_t) \cdot 10^3$}
    \\[0.1cm]
	& 
	True & FGMC & ED~ &
	True & FGMC & ED~ &
	True & FGMC & ED~ &
	True & FGMC & ED~ \\
	\specialrule{\cmidrulewidth}{0pt}{0pt}
    1.6 & 
		0.000 & -0.282 & \multicolumn{1}{c|}{$\times$} &
    	382.96 & 382.85 & \multicolumn{1}{c|}{$\times$} &
    	-7.927 & -8.016 & \multicolumn{1}{c|}{$\times$} &
    	42.962 & 43.089 & \multicolumn{1}{c|}{$\times$} \\
    1.2 & 
    	0.000 & -0.168 & \multicolumn{1}{c|}{$\times$} &
    	128.16 & 128.12 & \multicolumn{1}{c|}{$\times$} &
    	-2.121 & -2.141 & \multicolumn{1}{c|}{$\times$} &
    	37.805 & 37.945 & \multicolumn{1}{c|}{$\times$} \\
    0.8 & 
    	0.000 & -0.120 & 0.005 & 
    	64.727 &  64.695 &  64.644 & 
    	0.874 & 0.865 &  0.734 & 
    	35.798 &  35.971 &  35.965 \\
    0.4 & 
    	0.000 & -0.091 & -0.021 & 
    	40.517 & 40.480 & 40.469 & 
    	2.841 & 2.798 & 2.988 & 
    	36.108 & 36.060 & 36.448 \\
 	\bottomrule
\end{tabular}
}
\\[0.5cm]

\textbf{OU-NTS} \\[0.1cm]
\resizebox{\textwidth}{!}{
\begin{tabular}{|c|rrr|rrr|rrr|rrr|}
	\toprule
    \multirow{2}{*}{$\alpha$} & 
    \multicolumn{3}{c|}
    	{$c_1(X_t) \cdot 10^3$} & 
    \multicolumn{3}{c|}
    	{$c_2(X_t) \cdot 10^3$} & 
    \multicolumn{3}{c|}
    	{$c_3(X_t) \cdot 10^3$} & 
    \multicolumn{3}{c|}
    	{$c_4(X_t) \cdot 10^3$}
    \\[0.1cm]
	& 
	True & FGMC & ED~ &
	True & FGMC & ED~ &
	True & FGMC & ED~ &
	True & FGMC & ED~ \\
	\specialrule{\cmidrulewidth}{0pt}{0pt}
    0.8 & 
		0.000 & -0.083 & -0.065 & 
    	32.800 & 32.790 & 32.799 & 
    	0.000 & -0.000 & -0.007 & 
    	0.839 & 0.839 & 0.838 \\
    0.6 & 
		0.000 &  -0.084 & -0.086 & 
    	32.800 & 32.790 & 32.809 & 
    	0.000 & -0.003 & -0.007 & 
    	0.839 & 0.838 & 0.842 \\
    0.4 & 
		0.000 & -0.084 & -0.072 & 
    	32.800 & 32.791 & 32.809 & 
    	0.000 & -0.003 & -0.008 & 
    	0.839 & 0.837 & 0.836 \\
    0.2 & 
		0.000 & -0.084 & -0.075 & 
    	32.800 & 32.791 & 32.823 & 
    	0.000 & -0.003 & 0.000 & 
    	0.839 & 0.837 & 0.844 \\
    -1.0 & 
    	0.000 & -0.120 & \multicolumn{1}{c|}{$\times$} &
    	32.800 & 32.790 & \multicolumn{1}{c|}{$\times$} &
    	0.000 & -0.007 & \multicolumn{1}{c|}{$\times$} & 
    	0.839 & 0.841 & \multicolumn{1}{c|}{$\times$} \\
    -2.0 & 
    	0.000 & -0.121 & \multicolumn{1}{c|}{$\times$} &
    	32.800 & 32.790 & \multicolumn{1}{c|}{$\times$} &
    	0.000 & -0.006 & \multicolumn{1}{c|}{$\times$} &
    	0.839 & 0.842 & \multicolumn{1}{c|}{$\times$} \\
    \bottomrule
\end{tabular}
}
\\[0.5cm]

\textbf{NTS-OU} \\[0.1cm]
\resizebox{\textwidth}{!}{
\begin{tabular}{|c|rrr|rrr|rrr|rrr|}
	\toprule
    \multirow{2}{*}{$\alpha$} & 
    \multicolumn{3}{c|}
    	{$c_1(X_t) \cdot 10^3$} & 
    \multicolumn{3}{c|}
    	{$c_2(X_t) \cdot 10^3$} & 
    \multicolumn{3}{c|}
    	{$c_3(X_t) \cdot 10^3$} & 
    \multicolumn{3}{c|}
    	{$c_4(X_t) \cdot 10^3$}
    \\[0.1cm]
	& 
	True & FGMC & ED~ &
	True & FGMC & ED~ &
	True & FGMC & ED~ &
	True & FGMC & ED~ \\
	\specialrule{\cmidrulewidth}{0pt}{0pt}
   	0.8 & 
    	0.000 & -0.057 & -0.049 & 
    	14.183 & 14.176 & 14.175 & 
    	0.000 & 0.001 & -0.001 & 
    	0.725 & 0.726 & 0.719 \\
    0.6 & 
    	0.000 & -0.057 & -0.059 & 
    	14.183 & 14.177 & 14.171 & 
    	0.000 & -0.001 & -0.007 & 
    	0.725 & 0.724 & 0.715 \\
    0.4 & 
    	0.000 & -0.057 & -0.058 & 
    	14.183 & 14.178 & 14.196 & 
    	0.000 & -0.001 & -0.006 & 
    	0.725 & 0.724 & 0.729 \\
    0.2 & 
		0.000 & -0.057 & -0.039 & 
    	14.183 & 14.178 & 14.182 & 
    	0.000 & -0.002 & -0.002 & 
    	0.725 & 0.724 & 0.726 \\
 	\bottomrule
\end{tabular}
}
\caption{\small{Comparison of the first four cumulants (multiplied by 1000) obtained via the analytical expression (True), via the proposed method (FGMC) and via the Exact Decomposition (ED). We consider the TS and the NTS cases for both OU-Lévy and Lévy-OU processes, generating $10^7$ trials with a time horizon of $t=1$, and selecting different values for the stability parameter $\alpha$.
The symbol $\times$ indicates that no ED algorithm is available for that class of processes.
We notice that FGMC is extremely accurate in all cases; the obtained accuracy is at least at the fourth decimal digit for all cumulants.}}
\label{tab:cumulants_plain}
\end{table}

\begin{table}[!ht]
\centering
%\resizebox{\textwidth}{!}{
\begin{tabularx}{\textwidth}{|c|YY|YY||c|YY|YY|}
	\toprule
    \multirow{2}{*}{$\alpha$} & 
    \multicolumn{2}{c|}{\textbf{OU-TS}} & 
    \multicolumn{2}{c||}{\textbf{TS-OU}} &
    \multirow{2}{*}{$\alpha$} & 
    \multicolumn{2}{c|}{\textbf{OU-NTS}} &
    \multicolumn{2}{c|}{\textbf{NTS-OU}}
    \\[0.1cm]
	& 
	FGMC & ED &
	FGMC & ED & &
	FGMC & ED &
	FGMC & ED \\
	\hline
    1.6 & 
		0.87 & $\times$ &
		0.25 & $\times$ & 
    0.8 &
		0.74 & 5.44 & 
		0.39 & 3.46
    %\multirow{6}{*}{~ 0.17 ~}
    \\
    1.2 & 
		0.88 & $\times$ &
		0.27 & $\times$ &
    0.6 &
		0.72 & 13.75 &
		0.39 & 4.24 \\
    0.8 & 
		0.89 & 26.88 &
		0.29 & 10.11 &
    0.4 &
		0.73 & 15.15 &
		0.35 & 5.82 \\
    0.4 & 
		0.89 & 16.84 &
		0.31 & 5.02 &
    0.2 &
		0.75 & 15.85 &
		0.33 & 7.09 \\ 
    -1.0 & 
		0.59 & 4.81 &
    	\multicolumn{2}{c||}{\textit{\footnotesize{Not Defined}}}  &
    -1.0 &
     	0.69 & $\times$ &
    	\multicolumn{2}{c|}{\textit{\footnotesize{Not Defined}}} \\
    -2.0 & 
		0.58 & 3.01 &
    	\multicolumn{2}{c||}{\textit{\footnotesize{Not Defined}}} &
    -2.0 &
     	0.70 & $\times$ &
    	\multicolumn{2}{c|}{\textit{\footnotesize{Not Defined}}} \\ \bottomrule
\end{tabularx}
%}\\[0.5cm]

\caption{\small{
Computational times in seconds for simulating $10^7$ trials with our simulation technique (FGMC) and Exact Decomposition (ED). 
%Together with the four examined processes, simulation times for the Gaussian OU are provided as a benchmark. 
We simulate the processes on the yearly time horizon ($t=1$), maintaining consistent parameter values and ranges as in Table \ref{tab:cumulants_plain}. The symbol $\times$ denotes that, for that choice of $\alpha$, no ED algorithm is available in the literature, while the writing 
$\textit{\footnotesize{Not Defined}}$ indicates that the process is not defined.
The comparison reveals a significant disparity between the two approaches, with FGMC outpacing exact decomposition by at least one order of magnitude.
}}
\label{tab:computational_times_plain}
\end{table}

\begin{table}[!ht]
\centering
\textbf{Asymmetric OU-NTS} \\[0.1cm]
\resizebox{0.95\textwidth}{!}{
\begin{tabular}{|c|cc|cc|cc|cc||cc|}
	\toprule
    \multirow{2}{*}{$\alpha$} & 
    \multicolumn{2}{c|}{$c_1(X_t) \cdot 10^3$} & 
    \multicolumn{2}{c|}{$c_2(X_t) \cdot 10^3$} &  
    \multicolumn{2}{c|}{$c_3(X_t) \cdot 10^3$} & 
    \multicolumn{2}{c||}{$c_4(X_t) \cdot 10^3$} & 
    \multicolumn{2}{c|}{\phantom{------}Time [s]\phantom{------}} 
    \\[0.1cm]
	& 
	True & FGMC &
	True & FGMC &
	True & FGMC &
	True & FGMC &
	FGMC & ED \\
	\specialrule{\cmidrulewidth}{0pt}{0pt}
    0.8 & 
    	89.929 & 89.840 &
    	34.879 & 34.862 &
    	2.572 & 2.568 &
    	1.551 & 1.570 &
    	0.68 & $\times$ \\
    0.6 & 
    	89.929 & 89.840 &
    	34.879 & 34.863 &
    	2.451 & 2.446 &
    	1.234 & 1.237 &
    	0.65 & $\times$ \\
    0.4 & 
    	89.929 & 89.840 &
    	34.879 & 34.863 &
    	2.411 & 1.135 &
    	1.135 & 1.135 &
    	0.72 & $\times$ \\
    0.2 & 
    	89.929 & 89.840 &
    	34.879 & 34.863 &
    	2.391 & 2.385 &
    	1.087 & 1.086 &
    	0.71 & $\times$ \\
    -1.0 & 
    	89.929 & 89.801 &
    	34.879 & 34.859 &
    	2.355 & 2.345 &
    	1.002 & 1.001 &
    	0.58 & $\times$ \\
    -2.0 & 
    	89.929 & 89.800 &
    	34.879 & 34.861 &
    	2.347 & 2.340 &
    	0.983 & 0.986 &
    	0.56 & $\times$ \\ \bottomrule
\end{tabular}
}
\\[0.5cm]

\textbf{Asymmetric NTS-OU} \\[0.1cm]
\resizebox{0.95\textwidth}{!}{
\begin{tabular}{|c|cc|cc|cc|cc||cc|}
	\toprule
    \multirow{2}{*}{$\alpha$} & 
    \multicolumn{2}{c|}{$c_1(X_t) \cdot 10^3$} & 
    \multicolumn{2}{c|}{$c_2(X_t) \cdot 10^3$} &  
    \multicolumn{2}{c|}{$c_3(X_t) \cdot 10^3$} & 
    \multicolumn{2}{c||}{$c_4(X_t) \cdot 10^3$} &
    \multicolumn{2}{c|}{\phantom{------}Time [s]\phantom{------}} 
    \\[0.1cm]
	& 
	True & FGMC &
	True & FGMC &
	True & FGMC &
	True & FGMC &
	FGMC & ED \\
	\specialrule{\cmidrulewidth}{0pt}{0pt}
    0.8 & 
    	19.443 & 19.382 &
    	15.081 & 15.070 &
    	1.668 & 1.668 &
    	1.341 & 1.369 &
    	0.27 & $\times$ \\
    0.6 & 
    	19.443 & 19.382 &
    	15.081 & 15.071 &
    	1.590 & 1.586 &
    	1.067 & 1.072 &
    	0.28 & $\times$ \\
    0.4 & 
    	19.443 & 19.381 &
    	15.081 & 15.071 &
    	1.564 & 1.560 &
    	0.982 & 0.983 &
    	0.29 & $\times$ \\
    0.2 & 
    	19.443 & 19.381 &
    	15.081 & 15.072 &
    	1.551 & 1.547 &
    	0.940 & 0.940 &
    	0.32 & $\times$ \\ \bottomrule
\end{tabular}
}
\caption{\small{True and empirical cumulants (multiplied by 1000) obtained by generating $10^7$ samples with our technique (FGMC), together with the corresponding simulation times.
We simulate the processes considering a yearly horizon, i.e.\ $t=1$, using the same parameters as in Table \ref{tab:cumulants_plain} 
--~with the exception of the parameter $\theta$, which is chosen equal to 0.1, to ensure the asymmetry of the Lévy driver.
The symbol $\times$ indicates that no ED is available for that choice of process and parameters.
The results show that the FGMC obtains excellent results also in the asymmetric case, highlighting the great level of generality, speed and accuracy of the methodology.}}
\label{tab:cumulants_asymmetric}
\end{table}

\begin{table}[p]
\centering
\textbf{OU-TS} \\[0.1cm]
\resizebox{\textwidth}{!}{
\begin{tabular}{|c|cc|cc|cc|cc||cc|}
	\toprule
    \multirow{2}{*}{$\alpha$} &
    \multicolumn{2}{c|} {\phantom{-----}MAX [bp]\phantom{-----}} & 
    \multicolumn{2}{c|} {\phantom{----}RMSE [bp]\phantom{----}} &
    \multicolumn{2}{c|} {\phantom{----}MAPE [\%]\phantom{----}} &
    \multicolumn{2}{c||}{\phantom{------}$\overline{\textnormal{SD}}$ [bp]\phantom{------}} &
    \multicolumn{2}{c|}{Time [s]} \\[0.1cm]
	& 
	FGMC & ED &
	FGMC & ED &
	FGMC & ED &
	FGMC & ED &
	FGMC & ED \\
	\specialrule{\cmidrulewidth}{0pt}{0pt}
    1.6 & 
		2.21 & $\times$ &
		1.65 & $\times$ &
		0.03 & $\times$ &
		7.47 & $\times$ &
		0.99 & $\times$ \\
    1.2 & 
		2.72 & $\times$ &
		2.09 & $\times$ &
		0.06 & $\times$ &
		2.76 & $\times$ &
		0.98 & $\times$ \\
    0.8 & 
		2.16 & 1.69 &
		1.52 & 1.34 &
		0.06 & 0.05 &
		1.88 & 1.85 &
		0.97 & 26.28 \\
    0.4 & 
		1.64 & 1.45 &
		1.08 & 1.35 &
		0.05 & 0.08 &
		1.63 & 1.59 &
		0.97 & 17.35 \\
    -1.0 & 
		1.47 & 1.96 &
		1.32 & 1.59 &
		0.12 & 0.15 &
		2.01 & 2.20 &
		0.71 & 3.05 \\
    -2.0 & 
		1.65 & 3.99 &
		1.48 & 3.87 &
		0.12 & 0.32 &
		3.31 & 3.35 &
		0.61 & 2.13 \\
 \bottomrule
\end{tabular}
}\\[0.5cm]

\textbf{TS-OU} \\[0.1cm]
\resizebox{\textwidth}{!}{
\begin{tabular}{|c|cc|cc|cc|cc||cc|}
	\toprule
    \multirow{2}{*}{$\alpha$} &
    \multicolumn{2}{c|} {\phantom{-----}MAX [bp]\phantom{-----}} & 
    \multicolumn{2}{c|} {\phantom{----}RMSE [bp]\phantom{----}} &
    \multicolumn{2}{c|} {\phantom{----}MAPE [\%]\phantom{----}} &
    \multicolumn{2}{c||}{\phantom{------}$\overline{\textnormal{SD}}$ [bp]\phantom{------}} &
    \multicolumn{2}{c|}{Time [s]} \\[0.1cm]
	& 
	FGMC & ED &
	FGMC & ED &
	FGMC & ED &
	FGMC & ED &
	FGMC & ED \\
	\specialrule{\cmidrulewidth}{0pt}{0pt}
    1.6 & 
		2.53 & $\times$ &
		1.87 & $\times$ &
		0.07 & $\times$ &
		1.81 & $\times$ &
		0.49 & $\times$ \\
    1.2 & 
		1.46 & $\times$ &
		0.97 & $\times$ &
		0.05 & $\times$ &
		1.06 & $\times$ &
		0.47 & $\times$ \\
    0.8 & 
		0.75 & 0.36 &
		0.45 & 0.23 &
		0.03 & 0.02 &
		0.92 & 0.90 &
		0.42 & 13.39 \\
    0.4 & 
		1.53 & 1.53 &
		0.35 & 1.48 &
		0.04 & 0.29 &
		0.90 & 0.91 &
		0.45 & 5.51 \\
 \bottomrule
\end{tabular}
}\\[0.5cm]

\textbf{OU-NTS} \\[0.1cm]
\resizebox{\textwidth}{!}{
\begin{tabular}{|c|cc|cc|cc|cc||cc|}
	\toprule
    \multirow{2}{*}{$\alpha$} &
    \multicolumn{2}{c|} {\phantom{-----}MAX [bp]\phantom{-----}} & 
    \multicolumn{2}{c|} {\phantom{----}RMSE [bp]\phantom{----}} &
    \multicolumn{2}{c|} {\phantom{----}MAPE [\%]\phantom{----}} &
    \multicolumn{2}{c||}{\phantom{------}$\overline{\textnormal{SD}}$ [bp]\phantom{------}} &
    \multicolumn{2}{c|}{Time [s]} \\[0.1cm]
	& 
	FGMC & ED &
	FGMC & ED &
	FGMC & ED &
	FGMC & ED &
	FGMC & ED \\
	\specialrule{\cmidrulewidth}{0pt}{0pt}
    0.8 & 
		1.01 & 1.17 &
		0.66 & 0.77 &
		0.06 & 0.08 &
		0.37 & 0.37 &
		0.65 & 6.35 \\
    0.6 & 
		1.01 & 1.02 &
		0.66 & 0.66 &
		0.07 & 0.06 &
		0.37 & 0.37 &
		0.68 & 13.39 \\
    0.4 & 
		1.00 & 0.89 &
		0.66 & 0.55 &
		0.07 & 0.05 &
		0.37 & 0.37 &
		0.74 & 15.28 \\
    0.2 & 
		1.00 & 1.41 &
		0.66 & 1.03 &
		0.08 & 0.14 &
		0.37 & 0.37 &
		0.72 & 16.34 \\
    -1.0 & 
		0.49 & $\times$ &
		0.42 & $\times$ &
		0.08 & $\times$ &
		0.37 & $\times$ &
		0.60 & $\times$ \\
    -2.0 & 
		0.53 & $\times$ &
		0.46 & $\times$ &
		0.08 & $\times$ &
		0.37 & $\times$ &
		0.63 & $\times$ \\
 \bottomrule
\end{tabular}
}\\[0.5cm]

\textbf{NTS-OU} \\[0.1cm]
\resizebox{\textwidth}{!}{
\begin{tabular}{|c|cc|cc|cc|cc||cc|}
	\toprule
    \multirow{2}{*}{$\alpha$} &
    \multicolumn{2}{c|} {\phantom{-----}MAX [bp]\phantom{-----}} & 
    \multicolumn{2}{c|} {\phantom{----}RMSE [bp]\phantom{----}} &
    \multicolumn{2}{c|} {\phantom{----}MAPE [\%]\phantom{----}} &
    \multicolumn{2}{c||}{\phantom{------}$\overline{\textnormal{SD}}$ [bp]\phantom{------}} &
    \multicolumn{2}{c|}{Time [s]} \\[0.1cm]
	& 
	FGMC & ED &
	FGMC & ED &
	FGMC & ED &
	FGMC & ED &
	FGMC & ED \\
	\specialrule{\cmidrulewidth}{0pt}{0pt}
    0.8 & 
		0.71 & 0.58 &
		0.42 & 0.39 &
		0.04 & 0.07 &
		0.05 & 0.25 &
		0.34 & 3.73 \\
    0.6 & 
		0.69 & 0.75 &
		0.41 & 0.46 &
		0.05 & 0.05 &
		0.25 & 0.25 &
		0.37 & 4.58 \\
    0.4 & 
		0.67 & 0.71 &
		0.40 & 0.46 &
		0.05 & 0.06 &
		0.25 & 0.25 &
		0.41 & 5.75 \\
    0.2 & 
		0.65 & 0.84 &
		0.38 & 0.52 &
		0.05 & 0.06 &
		0.25 & 0.25 &
		0.39 & 7.23 \\
 \bottomrule
\end{tabular}
}\\[0.5cm]
\caption{\small{Comparison between European call pricing error for FGMC and Exact Decomposition (ED) methods. After computing the true prices with Lewis formula, we compute MAX, RMSE, MAPE and 
$\overline{\textnormal{SD}}$ for
30 call options (one-year maturity), with moneyness in the range $\sqrt{T}(-0.2,0.2)$; in addition, we report the computational time in seconds.
All results are obtained using $10^7$ trials, and
the symbol $\times$ indicates that no ED algorithm is available for that class of processes.
We notice that in all cases the error in pricing is of the order of one bp and that FGMC is always remarkably faster than the ED.
%\\[0.1cm]
%The benchmark time, i.e.\ that required in the Gaussian-OU case, is 0.23 s.
}}
\label{tab:eu_call}
\end{table}

\begin{table}[p]
\centering
\textbf{OU-TS} \\[0.1cm]
\resizebox{\textwidth}{!}{
\begin{tabular}{|c|ccc|ccc||ccc|ccc||cc|}
	\toprule
    \multirow{2}{*}{$\alpha$} &
    \multicolumn{3}{c|} {Price FGMC [\%]} &
    \multicolumn{3}{c||} {SD FGMC [bp]} &
    \multicolumn{3}{c|} {Price ED [\%]} &
    \multicolumn{3}{c||} {SD ED [bp]} &
    \multicolumn{2}{c|}{Time [s]}
    \\[0.1cm]
	& 
	OTM & ATM & ITM &
	OTM & ATM & ITM &
	OTM & ATM & ITM &
	OTM & ATM & ITM &
	FGMC & ED \\
	\specialrule{\cmidrulewidth}{0pt}{0pt}
	1.6 & 
        168.5 & 184.5 & 199.0 &
        22.6 & 22.8 & 22.8 &
		$\times$ & $\times$ & $\times$ &
		$\times$ & $\times$ & $\times$ &
		5.88 & $\times$ \\
	1.2 & 
        36.33 & 50.42 & 65.47 &
        3.18 & 3.42 & 3.58 &
		$\times$ & $\times$ & $\times$ &
		$\times$ & $\times$ & $\times$ &
		6.15 & $\times$ \\
	0.8 & 
        15.99 & 28.44 & 44.24 &
        1.62 & 1.87 & 2.03 &
        15.96 & 28.40 & 44.19 &
        1.63 & 1.87 & 2.03 &
		6.69 & 98.58 \\
	0.4 & 
        9.28 & 20.00 & 36.70 &
        1.23 & 1.45 & 1.59 &
        9.29 & 20.00 & 36.70 &
        1.25 & 1.47 & 1.60 &
		7.12 & 68.26 \\
	-1.0 & 
        6.11 & 10.84 & 31.16 &
        1.50 & 1.64 & 1.69 &
        6.14 & 10.87 & 31.19 &
        1.67 & 1.80 & 1.84 &
		2.89 & 13.74 \\
	-2.0 & 
        8.72 & 12.02 & 32.66 &
        3.12 & 3.21 & 3.23 &
        8.78 & 12.08 & 32.73 &
        2.90 & 3.00 & 3.03 &
		2.53 & 12.36 \\
	\bottomrule
\end{tabular}
}\\[0.5cm]

\textbf{TS-OU} \\[0.1cm]
\resizebox{\textwidth}{!}{
\begin{tabular}{|c|ccc|ccc||ccc|ccc||cc|}
	\toprule
    \multirow{2}{*}{$\alpha$} &
    \multicolumn{3}{c|} {Price FGMC [\%]} &
    \multicolumn{3}{c||} {SD FGMC [bp]} &
    \multicolumn{3}{c|} {Price ED [\%]} &
    \multicolumn{3}{c||} {SD ED [bp]} &
    \multicolumn{2}{c|}{Time [s]}
    \\[0.1cm]
	& 
	OTM & ATM & ITM &
	OTM & ATM & ITM &
	OTM & ATM & ITM &
	OTM & ATM & ITM &
	FGMC & ED \\
	\specialrule{\cmidrulewidth}{0pt}{0pt}
	1.6 & 
        14.97 & 26.17 & 40.18 &
        1.45 & 1.72 & 1.91 &
		$\times$ & $\times$ & $\times$ &
		$\times$ & $\times$ & $\times$ &
		4.81 & $\times$ \\
	1.2 &  
        3.23 & 10.46 & 25.73 &
        0.61 & 0.81 & 0.98 &
		$\times$ & $\times$ & $\times$ &
		$\times$ & $\times$ & $\times$ &
		4.96 & $\times$ \\
	0.8 & 
        2.45 & 10.44 & 31.03 &
        0.61 & 0.77 & 0.85 &
        2.45 & 10.44 & 31.03 &
        0.61 & 0.77 & 0.85 &
		5.38 & 89.79 \\
	0.4 & 
        1.63 & 5.46 & 27.04 &
        0.56 & 0.68 & 0.73 &
        1.68 & 5.75 & 27.43 &
        0.63 & 0.74 & 0.79 &
		6.07 & 77.05 \\
	\bottomrule
\end{tabular}
}\\[0.5cm]

\textbf{OU-NTS} \\[0.1cm]
\resizebox{\textwidth}{!}{
\begin{tabular}{|c|ccc|ccc||ccc|ccc||cc|}
	\toprule
    \multirow{2}{*}{$\alpha$} &
    \multicolumn{3}{c|} {Price FGMC [\%]} &
    \multicolumn{3}{c||} {SD FGMC [bp]} &
    \multicolumn{3}{c|} {Price ED [\%]} &
    \multicolumn{3}{c||} {SD ED [bp]} &
    \multicolumn{2}{c|}{Time [s]}
    \\[0.1cm]
	& 
	OTM & ATM & ITM &
	OTM & ATM & ITM &
	OTM & ATM & ITM &
	OTM & ATM & ITM &
	FGMC & ED \\
	\specialrule{\cmidrulewidth}{0pt}{0pt}
	0.8 & 
        0.27 & 6.17 & 26.11 &
        0.08 & 0.31 & 0.45 &
        0.26 & 6.19 & 26.10 &
        0.08 & 0.31 & 0.45 &
		5.25 & 47.92 \\
	0.6 & 
        0.27 & 6.18 & 26.11 &
        0.08 & 0.31 & 0.45 &
        0.27 & 6.18 & 26.10 &
        0.08 & 0.31 & 0.45 &
		5.99 & 46.85 \\
	0.4 & 
        0.27 & 6.17 & 26.11 &
        0.08 & 0.31 & 0.45 &
        0.28 & 6.16 & 26.10 &
        0.08 & 0.31 & 0.45 &
		6.32 & 45.76 \\
	0.2 & 
        0.28 & 6.17 & 26.11 &
        0.08 & 0.31 & 0.45 &
        0.28 & 6.16 & 26.10 &
        0.08 & 0.31 & 0.45 &
		6.28 & 73.62 \\
	-1.0 & 
        0.28 & 6.16 & 26.11 &
        0.08 & 0.31 & 0.45 &
		$\times$ & $\times$ & $\times$ &
		$\times$ & $\times$ & $\times$ &
		5.16 & $\times$ \\
	-2.0 & 
        0.28 & 6.15 & 26.11 &
        0.08 & 0.31 & 0.45 &
		$\times$ & $\times$ & $\times$ &
		$\times$ & $\times$ & $\times$ &
		5.69 & $\times$ \\
	\bottomrule
\end{tabular}
}\\[0.5cm]

\textbf{NTS-OU} \\[0.1cm]
\resizebox{\textwidth}{!}{
\begin{tabular}{|c|ccc|ccc||ccc|ccc||cc|}
	\toprule
    \multirow{2}{*}{$\alpha$} &
    \multicolumn{3}{c|} {Price FGMC [\%]} &
    \multicolumn{3}{c||} {SD FGMC [bp]} &
    \multicolumn{3}{c|} {Price ED [\%]} &
    \multicolumn{3}{c||} {SD ED [bp]} &
    \multicolumn{2}{c|}{Time [s]}
    \\[0.1cm]
	& 
	OTM & ATM & ITM &
	OTM & ATM & ITM &
	OTM & ATM & ITM &
	OTM & ATM & ITM &
	FGMC & ED \\
	\specialrule{\cmidrulewidth}{0pt}{0pt}
	0.8 & 
        0.07 & 3.74 & 25.26 &
        0.04 & 0.20 & 0.29 &
        0.07 & 3.74 & 25.25 &
        0.05 & 0.20 & 0.29 &
		5.40 & 42.47 \\
	0.6 &  
        0.07 & 3.66 & 25.26 &
        0.04 & 0.20 & 0.29 &
        0.07 & 3.66 & 25.26 &
        0.04 & 0.20 & 0.29 &
		5.70 & 46.87 \\
	0.4 & 
        0.07 & 3.61 & 25.26 &
        0.04 & 0.20 & 0.29 &
        0.07 & 3.60 & 25.25 &
        0.04 & 0.20 & 0.29 &
		6.24 & 49.97 \\
	0.2 & 
        0.07 & 3.61 & 25.26 &
        0.04 & 0.20 & 0.29 &
        0.07 & 3.56 & 25.26 &
        0.04 & 0.21 & 0.29 &
		7.22 & 51.77 \\
	\bottomrule
\end{tabular}
}\\[0.5cm]
\caption{\small{Pricing Asian options using FGMC and Exact Decomposition (ED) methods. We evaluate an arithmetic average Asian option with two-year ($T=2$) expiry and monthly monitoring.
We consider 3 relevant levels of moneyness 
(OTM, ATM, ITM), corresponding to $\rchi=0.2 \sqrt{T}, \rchi=0, \rchi=-0.2 \sqrt{T}$ respectively, and we compute the price and the MC standard deviation for different values of the stability parameters $\alpha$. As in the previous experiments, estimates are obtained simulating $10^7$ paths, using the same parameters as in Table \ref{tab:cumulants_plain}, along with the simulation times in seconds for both methods.
The symbol $\times$ denotes that, for that choice of $\alpha$, no ED is available in the literature. 
In all cases prices are consistent across the two methods, but FGMC is always remarkably faster than the ED.
%\\[0.1cm]
%The benchmark time, i.e.\ that required in the Gaussian-OU case, is 2.45 s.
}}
\label{tab:asian_call}
\end{table}

\clearpage\newpage

\bibliography{FGMCBiblio.bib}

@book{abramowitz1988,
  title={{Handbook of mathematical functions with formulas, graphs, and mathematical tables}},
  author={Abramowitz, Milton and Stegun, Irene A and Romer, Robert H},
  year={1988},
  publisher={American Association of Physics Teachers}
}

@book{ahlfors1979complex,
  title={Complex analysis},
  author={Ahlfors, Lars Valerian},
  volume={3},
  year={1979},
  publisher={McGraw-Hill New York}
}

@article{azzone2023,
  title={{A fast Monte Carlo scheme for additive processes and option
  pricing}},
  author={Azzone, Michele and Baviera, Roberto},
  journal={Computational Management Science},
  volume={20},
  number={1},
  pages={31},
  year={2023},
  publisher={Springer}
}

@article{ballotta2014,
  title={{Monte Carlo simulation of the CGMY process and option pricing}},
  author={Ballotta, Laura and Kyriakou, Ioannis},
  journal={Journal of Futures Markets},
  volume={34},
  number={12},
  pages={1095--1121},
  year={2014},
  publisher={Wiley Online Library}
}

@article{barndorff2001,
  title={{Non-Gaussian Ornstein--Uhlenbeck-based models and some of their
  uses in financial economics}},
  author={Barndorff-Nielsen, Ole E and Shephard, Neil},
  journal={Journal of the Royal Statistical Society: Series B (Statistical
  Methodology)},
  volume={63},
  number={2},
  pages={167--241},
  year={2001},
  publisher={Wiley Online Library}
}

@article{barndorff2001normal,
	title = {Normal modified stable processes},
	journal = {Theory of Probability and Mathematical Statistics},
	year = {2001},
	pages = {1-19},
	author = {Barndorff-Nielsen, Ole E. and Neil Shephard},
	volume = {65}
}

@misc{bavieramassaria2026, 
    title={{Smile asymptotics for Bachelier implied volatility}}, 
    author={Roberto Baviera and Michele Domenico Massaria}, 
    year={2026}, 
    note={arXiv preprint arXiv:2506.08067} 
}

@book{benth2008,
  title={{Stochastic modelling of electricity and related markets}},
  author={Benth, Fred Espen and Benth, Jurate Saltyte and Koekebakker,
  Steen},
  year={2008},
  publisher={World Scientific}
}

@article{benth2018lavagnini,
  title={Stochastic modeling of wind derivatives in energy markets},
  author={Benth, Fred Espen and Di Persio, Luca and Lavagnini, Silvia},
  journal={Risks},
  volume={6},
  number={2},
  pages={56},
  year={2018},
  publisher={MDPI}
}

@article{benth2019mean,
  title={{Mean-reverting additive energy forward curves in a Heath-Jarrow-Morton framework}},
  author={Benth, Fred Espen and Piccirilli, Marco and Vargiolu, Tiziano},
  journal={Mathematics and Financial Economics},
  volume={13},
  number={4},
  pages={543--577},
  year={2019},
  publisher={Springer}
}

@article{bianchi2017,
  title={{Tempered stable Ornstein--Uhlenbeck processes: A practical view}},
  author={Bianchi, Michele Leonardo and Rachev, Svetlozar T and Fabozzi, Frank J},
  journal={Communications in Statistics-Simulation and Computation},
  volume={46},
  number={1},
  pages={423--445},
  year={2017},
  publisher={Taylor \& Francis}
}

@article{brix1999,
  title={{Generalized gamma measures and shot-noise Cox processes}},
  author={Anders {Brix}},
  journal={{Advances in Applied Probability}},
  volume={31},
  number={4},
  pages={929--953},
  year={1999},
  publisher={Cambridge University Press}
}

@article{carr2002,
  title={The fine structure of asset returns: An empirical investigation},
  author={Carr, Peter and Geman, H{\'e}lyette and Madan, Dilip B and Yor, Marc},
  journal={The Journal of Business},
  volume={75},
  number={2},
  pages={305--332},
  year={2002},
  publisher={JSTOR}
}

@article{chen2012,
  title={{Simulating Lévy processes from their characteristic functions and financial applications}},
  author={Chen, Zisheng and Feng, Liming and Lin, Xiong},
  journal={ACM Transactions on Modeling and Computer Simulation (TOMACS)},
  volume={22},
  number={3},
  pages={1--26},
  year={2012},
  publisher={ACM New York, NY, USA}
}

@book{comtet2012advanced,
  title={{Advanced Combinatorics: The art of finite and infinite expansions}},
  author={Comtet, Louis},
  year={2012},
  publisher={Springer Science \& Business Media}
}

@book{cont2003,
	author = {Cont, R and Tankov, P},
	publisher = {Chapman \& Hall/CRC},
	title = {{Financial Modelling with Jump Processes}},
	year = {2003}
}

@article{cooley1965,
  title={{An algorithm for the machine calculation of complex Fourier series}},
  author={Cooley, James W and Tukey, John W},
  journal={Mathematics of computation},
  volume={19},
  number={90},
  pages={297--301},
  year={1965}
}

@article{cummins2017,
  title={{Gas storage valuation under Lévy processes using the Fast Fourier Transform}},
  author={Cummins, Mark and Kiely, Greg and Murphy, Bernard},
  journal={Journal of Energy Markets},
  volume={10},
  number={4},
  pages={43--86},
  year={2017}
}

@article{devroye2009,
  title={{Random variate generation for exponentially and polynomially tilted stable distributions}},
  author={Luc {Devroye}},
  journal={ACM Transactions on Modeling and Computer Simulation (TOMACS)},
  volume={19},
  number={4},
  pages={1--20},
  year={2009},
  publisher={ACM New York, NY, USA}
}

@article{fusai2008,
  title={{Analytical pricing of discretely monitored Asian-style options: Theory and application to commodity markets}},
  author={Fusai, Gianluca and Marena, Marina and Roncoroni, Andrea},
  journal={Journal of Banking \& Finance},
  volume={32},
  number={10},
  pages={2033--2045},
  year={2008},
  publisher={Elsevier}
}

@article{flury1990,
  title={Acceptance--rejection sampling made easy},
  author={Flury, Bernard D},
  journal={Siam Review},
  volume={32},
  number={3},
  pages={474--476},
  year={1990},
  publisher={SIAM}
}

@article{gil1951,
  title={{Note on the inversion theorem}},
  author={Gil-Pelaez, J},
  journal={Biometrika},
  volume={38},
  number={3-4},
  pages={481--482},
  year={1951},
  publisher={Oxford University Press}
}

@article{hambly2009,
	author = {Hambly, B. and Howison, S. and Kluge, Tino},
	year = {2009},
	month = {12},
	pages = {937-949},
	title = {Modeling Spikes and Pricing Swing Options in Electricity Markets},
	volume = {9},
	journal = {Quantitative Finance}
}

@article{hofert2011,
  title={{Sampling exponentially tilted stable distributions}},
  author={Marius {Hofert}},
  journal={{ACM Transactions on Modeling and Computer Simulation (TOMACS)}},
  volume={22},
  number={1},
  pages={1--11},
  year={2011},
  publisher={ACM New York, NY, USA}
}

@article{kawai2011,
  title={{Exact discrete sampling of finite variation tempered stable Ornstein--Uhlenbeck processes}},
  author={Kawai, Reiichiro and Masuda, Hiroki},
  year={2011},
  pages={279--300},
  publisher={Walter de Gruyter GmbH \& Co. KG},
  journal={Monte Carlo Methods and Applications}
}

@article{kawai2012,
  title={{Infinite variation tempered stable Ornstein--Uhlenbeck processes with discrete observations}},
  author={Kawai, Reiichiro and Masuda, Hiroki},
  journal={Communications in Statistics-Simulation and Computation},
  volume={41},
  number={1},
  pages={125--139},
  year={2012},
  publisher={Taylor \& Francis}
}

@article{kuchler2013,
  title={{Tempered stable distributions and processes}},
  author={K{\"u}chler, Uwe and Tappe, Stefan},
  journal={Stochastic Processes and their Applications},
  volume={123},
  number={12},
  pages={4256--4293},
  year={2013},
  publisher={Elsevier}
}

@article{latini2019mean,
  title={{Mean-reverting no-arbitrage additive models for forward curves in energy markets}},
  author={Latini, Luca and Piccirilli, Marco and Vargiolu, Tiziano},
  journal={Energy Economics},
  volume={79},
  pages={157--170},
  year={2019},
  publisher={Elsevier}
}

@article{lee2004,
  title={{Option pricing by transform methods: extensions, unification and
  error control}},
  author={Lee, Roger W.},
  journal={Journal of Computational Finance},
  volume={7},
  number={3},
  pages={51--86},
  year={2004},
  publisher={RISK PUBLICATIONS}
}

@article{lewis2001simple,
  title={A simple option formula for general jump-diffusion and other exponential L{\'e}vy processes},
  author={Lewis, Alan L},
  journal={Available at SSRN 282110},
  year={2001}
}

@article{lukacs1969,
  title={{A characterization of stable processes}},
  author={Lukacs, Eugene},
  journal={Journal of Applied Probability},
  volume={6},
  number={2},
  pages={409--418},
  year={1969},
  publisher={Cambridge University Press}
}

@book{lukacs1970,
  author    = {E. Lukacs},
  title     = {Characteristic Functions},
  edition   = {2nd},
  year      = {1970},
  publisher = {Griffin},
  address   = {London}
}

@article{lukacs1972,
  title={{A survey of the theory of characteristic functions}},
  author={Lukacs, Eugene},
  journal={Advances in Applied Probability},
  volume={4},
  number={1},
  pages={1--37},
  year={1972},
  publisher={Cambridge University Press}
}

@book{press2007numerical,
  title={{Numerical recipes: The art of scientific computing}},
  author={Press, William H and Teukolsky, Saul A and Vetterling, WT and Flannery, Brian P},
  year={2007},
  edition = {Third},
  publisher={Cambridge University Press}
}

@book{quarteroni2006,
  title={Numerical mathematics},
  author={Quarteroni, Alfio and Sacco, Riccardo and Saleri, Fausto},
  volume={37},
  year={2006},
  publisher={Springer Science \& Business Media}
}

@article{qu2021,
  title={{Exact simulation of Ornstein--Uhlenbeck tempered stable processes}},
  author={Qu, Yan and Dassios, Angelos and Zhao, Hongbiao},
  journal={Journal of Applied Probability},
  volume={58},
  number={2},
  pages={347--371},
  year={2021},
  publisher={Cambridge University Press}
}

@book{reed1975,
  title={{Methods of Modern Mathematical Physics, II: Fourier analysis, self-adjointness}},
  author={Reed, Michael and Simon, Barry},
  volume={2},
  year={1975},
  publisher={Elsevier}
}

@article{sabino2021gamma,
  title={{Gamma-related Ornstein--Uhlenbeck processes and their simulation}},
  author={Sabino, Piergiacomo and Cufaro Petroni, Nicola},
  journal={Journal of Statistical Computation and Simulation},
  volume={91},
  number={6},
  pages={1108--1133},
  year={2021},
  publisher={Taylor \& Francis}
}

@article{sabino2020ouvg,
  title={{Exact simulation of Variance Gamma-related OU processes: application to the pricing of energy derivatives}},
  author={Sabino, Piergiacomo},
  journal={Applied Mathematical Finance},
  volume={27},
  number={3},
  pages={207--227},
  year={2020},
  publisher={Taylor \& Francis}
}

@article{sabino2022fast,
  title={{Fast simulation of Tempered Stable Ornstein--Uhlenbeck processes}},
  author={Sabino, Piergiacomo and Cufaro Petroni, Nicola},
  journal={Computational Statistics},
  volume={37},
  number={5},
  pages={2517--2551},
  year={2022},
  publisher={Springer}
}

@article{sabino2022cgmy,
  title={{Pricing energy derivatives in markets driven by Tempered Stable 		and CGMY processes of Ornstein--Uhlenbeck type}},
  author={Sabino, Piergiacomo},
  journal={Risks},
  volume={10},
  number={8},
  pages={148},
  year={2022},
  publisher={MDPI}
}

@article{sabino2022exact,
  title={{Exact simulation of Normal Tempered Stable processes of OU type with applications}},
  author={Sabino, Piergiacomo},
  journal={Statistics and Computing},
  volume={32},
  number={5},
  pages={81},
  year={2022},
  publisher={Springer}
}

@article{sabino2023normal,
  title={{Normal Tempered Stable processes and the pricing of energy derivatives}},
  author={Sabino, Piergiacomo},
  journal={SIAM Journal on Financial Mathematics},
  volume={14},
  number={1},
  pages={99--126},
  year={2023},
  publisher={SIAM}
}

@book{sato1999,
  title={{L{\'e}vy processes and infinitely divisible distributions}},
  author={Sato, Ken Iti},
  volume={68},
  year={1999},
  publisher={Cambridge university press}
}

@book{stein2010complex,
  title={{Complex analysis}},
  author={Stein, Elias M and Shakarchi, Rami},
  volume={2},
  year={2010},
  publisher={Princeton University Press}
}

@article{taufer2009,
  title={{Simulation of Lévy-driven Ornstein-Uhlenbeck processes with given marginal distribution}},
  author={Taufer, Emanuele and Leonenko, Nikolai},
  journal={Computational Statistics \& Data Analysis},
  volume={53},
  number={6},
  pages={2427--2437},
  year={2009},
  publisher={Elsevier}
}

@article{wendel1961,
  title={{The non-absolute convergence of Gil-Pelaez' inversion integral}},
  author={Wendel, JG},
  journal={The Annals of Mathematical Statistics},
  volume={32},
  number={1},
  pages={338--339},
  year={1961},
  publisher={Institute of Mathematical Statistics}
}

@article{wolfe1982,
  title={{On a continuous analogue of the stochastic difference equation 
  $X_n= \rho X_{n-1}+ B_n$}},
  author={Wolfe, Stephen James},
  journal={Stochastic Processes and their applications},
  volume={12},
  number={3},
  pages={301--312},
  year={1982},
  publisher={Elsevier}
}

@article{zhang2008exact,
  title={{Exact simulation of IG-OU processes}},
  author={Zhang, Shibin and Zhang, Xinsheng},
  journal={Methodology and computing in applied probability},
  volume={10},
  number={3},
  pages={337--355},
  year={2008},
  publisher={Springer}
}

@article{zhang2009,
  title={{On the transition law of tempered stable Ornstein--Uhlenbeck 
  processes}},
  author={Zhang, Shibin and Zhang, Xinsheng},
  journal={Journal of Applied Probability},
  volume={46},
  number={3},
  pages={721--731},
  year={2009},
  publisher={Cambridge University Press}
}

@article{zhang2011transition,
  title={{Transition law-based simulation of generalized inverse Gaussian Ornstein--Uhlenbeck processes}},
  author={Zhang, Shibin},
  journal={Methodology and Computing in Applied Probability},
  volume={13},
  number={3},
  pages={619--656},
  year={2011},
  publisher={Springer}
}

\vspace{1cm}
\newpage
\appendix

\bigskip
\noindent
%{\LARGE\bfseries{Appendices}}

\renewcommand{\theequation}{\thesection.\arabic{equation}}
\renewcommand{\thetheorem}{\thesection.\arabic{theorem}}

\clearpage \newpage

\setcounter{equation}{0}
\setcounter{theorem}{0}

\section{Proofs} \label{sec:appendixA}

In this Appendix, we collect the proofs for the results presented in the main text. 
%as well as for Appendix~\ref{sec:appendix_NTS}. 
We first prove a useful proposition, which is instrumental to the proofs of Propositions \ref{prop:oulevystrip} and \ref{prop:levyoustrip}.

\vspace{-0.3\baselineskip}
\begin{proposition} \label{prop:append_infdiv}
Let $X$ be an infinitely divisible r.v.\ with CF $\phi_X$ and LCF $\Psi_X$. Then, the analyticity strips of $\phi_X$ and $\Psi_X$ coincide.
\end{proposition}
\vspace{-0.9\baselineskip}
\begin{proof}
Let $\mathcal{S} := \{u \in \mathbb{C} : -p_+ < \text{Im}(u) < p_-\} \neq \varnothing$ be the analyticity strip of the CF $\phi_X$.  
Since $X$ is infinitely divisible, $\phi_X$ has no zeros in $\mathcal{S}$ \citep[cf.][Theorem 5.4a, p.21]{lukacs1972}.

The strip $\mathcal{S}$ is simply connected, and $\phi_X$ is holomorphic and non-vanishing there.
By a standard result of complex analysis
\citep[cf., e.g.,][Theorem 16, p.143]{ahlfors1979complex},
the complex logarithm of $\phi_X$ admits a single-valued holomorphic branch on 
$\mathcal{S}$.
Hence, the LCF $\Psi_X(u) := \log \phi_X(u)$ 
--~originally defined only on the real axis~-- can be analytically continued at least to all of $\mathcal{S}$ \citep[see also][Theorem 25.17, p.165]{sato1999};
therefore, the analyticity strip of $\Psi_X$, denoted as $\mathcal{S}'$, satisfies 
$\mathcal{S}' \supseteq \mathcal{S}$.

Conversely, since the complex exponential $e^z$ is an entire function, $\phi_X(u) = e^{\Psi_X(u)}$ is analytic at every point where $\Psi_X$ is analytic, implying $\mathcal{S} \supseteq \mathcal{S}'$.  
Together, the two inclusions yield $\mathcal{S} = \mathcal{S}'$.
\end{proof}

% =============================
% 	PROPOSITION 1
% =============================

\noindent
{\sffamily\bfseries Proof of Proposition \ref{prop:oulevystrip}}

\noindent
Fix $t>0$, and let $\mathcal{S}_L$ denote the analyticity strip of the characteristic exponent $\psi_L$.
We assume the strip $\mathcal{S}_L$ to be bounded; the unbounded case follows immediately.
Define
\begin{equation*}
	\zeta(u,s) := 
	\psi_L( u e^{-bs} ) \,,
	\quad\quad \textnormal{with} ~ (u,s) \in \mathcal{S}_L \times [0,t] \,.
\end{equation*}
\noindent
The function $\zeta(u,s)$ satisfies two properties:
\begin{itemize}[topsep=3pt]
	\setlength\itemsep{0pt}
    \item[(i)] For each fixed $s \in [0,t]$, we have $0 < e^{-bs} \leq 1$, so the function $u \mapsto u e^{-bs}$ maps $\mathcal{S}_L$ into itself. Since $\psi_L$ is holomorphic on $\mathcal{S}_L$, it follows that $\zeta(\bigcdot,s)$ is holomorphic on $\mathcal{S}_L$.
    \item[(ii)] The function $\zeta(u,s)$ is jointly continuous in $(u,s)$ on $\mathcal{S}_L \times [0,t]$, because it is a composition of continuous functions.
\end{itemize}

\noindent
By Theorem 5.4 in \citet[][p.~56]{stein2010complex}, these two conditions guarantee that the integral
\[
\Psi_Z(u,t) = \int_0^t \zeta(u,s)\,\dd s = \int_0^t \psi_L( u e^{-bs} ) \, \dd s \,
\]
is holomorphic for every $u \in \mathcal{S}_L$. 

\smallskip
Conversely, for any singularity of $\psi_L$ on the boundary of $\mathcal{S}_L$
there exist points $\tilde{u} \in \mathbb{C}$ arbitrarily close to the boundary, but outside $\mathcal{S}_L$, such that 
the set $\{ \tilde{u} e^{-bs} : s \in [0,t] \}$ intersects this singularity. 
For such $\tilde{u}$, the integral is ill-defined,
entailing that $\Psi_Z(\bigcdot,t)$ cannot be analytically continued outside $\mathcal{S}_L$.

\smallskip
Hence, the LCF $\Psi_Z$ is holomorphic in exactly the same strip as the characteristic exponent $\psi_L$.
Since both $Z_t$ and $L_t$ are infinitely divisible, Proposition \ref{prop:append_infdiv} allows us to translate this analyticity result into the CFs $\phi_Z(u,t)$ and $\phi_L(u)$, concluding the proof.
$\hfill\square$

%Here, we have focused on the general case where $\mathcal{S}_L$ is bounded; the unbounded case follows directly.
%$\hfill\square$

% =============================
% 	PROPOSITION 2
% =============================

\bigskip
\noindent
{\sffamily\bfseries Proof of Proposition \ref{prop:levyoustrip}}

\noindent
Fix $t>0$, and let $\mathcal{S}_X$ denote the analyticity strip of the characteristic exponent $\psi_X$.
We assume the strip $\mathcal{S}_X$ to be bounded; the unbounded case follows immediately.
%Let $\mathcal{S}_X$ denote the analyticity strip of the characteristic exponent $\psi_X$, and fix $t>0$. 
Since $0 < e^{-bt} \leq 1$, for every $u \in \mathcal{S}_X$ both terms $\psi_X(u)$ and $\psi_X(u e^{-bt})$ in \eqref{eq:levyou-relationship} are holomorphic.
Hence, $\Psi_Z(u,t)$ is holomorphic for all $u \in \mathcal{S}_X$.

Conversely, let $\tilde{u} \in \mathbb{C}$ be a singularity of $\psi_X$ on the boundary of $\mathcal{S}_X$. 
At $\tilde{u}$, the first term $\psi_X(\tilde{u})$ is non-holomorphic, 
while the second term $\psi_X(\tilde{u} e^{-bt})$ remains holomorphic since 
$\tilde{u} e^{-bt} \in \mathcal{S}_X$. Thus, $\Psi_Z(\tilde{u},t)$ is non-holomorphic on the boundary of $\mathcal{S}_X$.

\smallskip
Therefore, the LCF $\Psi_Z(u,t)$ is holomorphic in exactly the same strip as $\psi_X$.
Since both $Z_t$ and $X$ are infinitely divisible, Proposition \ref{prop:append_infdiv} allows us to translate this result into the respective CFs.
$\hfill\square$

% =============================
% 	PROPOSITION 3
% =============================

\bigskip\noindent
{\color{my_colour}
{\sffamily\bfseries Proof of Proposition \ref{prop:abs_continuity_of_Z_t}}

\noindent
Since the Lévy driver $L_t$ is a pure-jump process with Lévy density $\nu_L(\bigcdot)$ 
(i.e.\ its Lévy measure is absolutely continuous with respect to the Lebesgue measure), the Lévy measure of $Z_t$ is also absolutely continuous for every $t>0$, with Lévy density 
$\nu_Z(\bigcdot, t)$ given by \eqref{eq:oulevymeasure}. For any such $t$, two cases are possible:
\begin{itemize}[topsep=3pt]
    \setlength\itemsep{0pt}
    \item[(i)] $Z_t$ has infinite-activity, i.e.\ 
    $\int_{\mathbb{R}} \, \nu_Z(x, t) \dd{x} = \infty$. In this case, a result of \citet[Theorem 27.7, p.177]{sato1999} ensures the absolute continuity of the law of the process at any finite time.
    \item[(ii)] $Z_t$ has finite-activity, i.e.\
    $\int_{\mathbb{R}} \, \nu_Z(x, t) \dd{x} < \infty$. In this case, $Z_t$ is equivalent in law to a compound Poisson r.v.\ whose i.i.d.\ jumps are absolutely continuous, as discussed in detail in Section \ref{subsec:cpsimul}.
    A mixture of absolutely continuous laws remains absolutely continuous; therefore, conditionally on at least one jump occurring, $Z_t$ is absolutely continuous.
    Unconditionally, $Z_t$ has a mixed distribution with an atom at zero (if no jumps occur) and an absolutely continuous part elsewhere. $\hfill\square$
\end{itemize}
}

% =============================
% 	PROPOSITION 4
% =============================

\bigskip
\noindent
{\sffamily\bfseries Proof of Proposition \ref{prop:numerical_error1}}

\noindent
The approximation error $\mathcal{E}(x)$ is given by the absolute difference between the exact CDF $P(x)$ and its approximation $\widehat{P}(x)$ in \eqref{eq:Sigma_N_equation}.
Letting $\widehat{P}_\infty(x)$ denote the limit of $\widehat{P}(x)$ as the number of terms $N$ in the summation goes to infinity, we write
\begin{equation*}
	\mathcal{E}(x) :=
	\abs{P(x) - \widehat{P}(x)} 
	\leq
	\abs{P(x) - \widehat{P}_\infty(x)}  + 
	\abs{\widehat{P}_\infty(x) - \widehat{P}(x)} \,,
\end{equation*}
decomposing the approximation error into a discretisation term and a truncation term, respectively.

\smallskip \noindent
For the truncation error, we distinguish two cases:
\begin{itemize}[topsep=3pt]
\item[(i)] if the decay of the CF modulus is exponential, one has
\begin{equation*}\begin{gathered}
	\abs{\widehat{P}_\infty(x) - \widehat{P}(x)} =
	\frac{he^{-a x}}{\pi}
	\abs{
	\sum_{n=N}^{\infty}
		\frac{\phi (\left(n+\frac{1}{2} \right) h - ia)}
			{i\left(n+\frac{1}{2} \right) h + a}
		e^{-ix \left(n+\frac{1}{2} \right) h}
	}
	\, \leq \,
	\frac{he^{-a x}}{\pi}
	\sum_{n=N}^{\infty}
	\abs{
		\frac{\phi (\left(n+\frac{1}{2} \right) h - ia)}
			{i\left(n+\frac{1}{2} \right) h + a}
	}
	\leq
	\\
	\overset{(A)}{\leq}
	\,\,
	\frac{B h e^{-a x}}{\pi}
	\sum_{n=N}^{\infty}
	\frac{e^{-\ell \left(\left( n+\frac{1}{2} \right) h\right)^\omega}}
		{\left( n+\frac{1}{2} \right) h}
	\,\,
	\overset{(B)}{\leq}
	\,\,
	\frac{B e^{-a x}}{\pi \omega}
	\int_{\ell(Nh)^{\omega}}^{\infty}
		\frac{e^{-y}}{y}
		\dd{y}
	\,\,
	\overset{(C)}{\leq}
	\,\,
	\frac{B e^{-a x}}{\pi \omega} 
	\frac{e^{-\ell(Nh)^{\omega}}}{\ell(Nh)^{\omega}} \,.
\end{gathered}\end{equation*}
Inequality $(A)$ is due to the fact that 
$\abs{i\left(n+1/2 \right) h + a} \geq \left(n+1/2 \right) h$; 
inequality $(B)$ exploits the eventual convexity of the function
${e^{-\ell x^{\omega}}}/{x}$
to bound the summation with the integral;
inequality $(C)$ is given by (5.1.20) in
\citet{abramowitz1988}, p.229.
\item[(ii)]
if the decay of the CF modulus is power-law, one has
\begin{equation*}\begin{gathered}
	\abs{\widehat{P}_\infty(x) - \widehat{P}(x)} =
	\frac{h	e^{-a x}}{\pi}
	\abs{
	\sum_{n=N}^{\infty}
		\frac{\phi (\left(n+\frac{1}{2} \right) h-ia)}
			{i\left(n+\frac{1}{2} \right) h + a}
		e^{-ix \left(n+\frac{1}{2} \right) h}
	}
	\, \leq \,
	\frac{he^{-a x}}{\pi}
	\sum_{n=N}^{\infty}
	\abs{
		\frac{\phi (\left(n+\frac{1}{2} \right) h - ia)}
			{i\left(n+\frac{1}{2} \right) h + a}
	}
	\leq
	\\
	\overset{(A)}{\leq}
	\,\,
	\frac{B h e^{-a x}}{\pi}
	\sum_{n=N}^{\infty}
		 \left( 
		 	\left(
				n+\frac{1}{2} 
			\right) 
			h 
		\right)
	^{ -\omega - 1}
	\,\,
	\overset{(B)}{\leq}
	\,\,
	\frac{B e^{-a x}}{\pi}
	\int_{Nh}^{\infty}
		x^{ -\omega - 1}
	\dd{x}
	=
	\frac{B e^{-a x}}{\pi \omega}
		(Nh)^{-\omega} \,.
\end{gathered}\end{equation*}
Inequality $(A)$ is due to the fact that
$\abs{i\left(n+1/2 \right) h + a} \geq \left(n+1/2 \right) h$, as in the previous case, and
inequality $(B)$ exploits the convexity of the function
$x^{-\omega-1}$.
\end{itemize}

\smallskip
\noindent
For the discretisation error, we rely on Theorems 6.2 and 6.5 in \citet{lee2004}.
For our purposes, the relevant case is Lee's $G_3$ payoff with $b_0 = 1$ and $b_1=0$, and we set $p=2a$ to obtain the appropriate form of the discretisation bound.
$\hfill\square$

\bigskip
\noindent
{\sffamily\bfseries Proof of Lemma \ref{prop:ounts_decay}}

\noindent
In the IA case, we consider the complex term appearing inside the integral in \eqref{eq:ounts_psi}, evaluated at $u-ia$:
 \begin{equation*}
	w :=
	\frac{1}{2} \sigma^2 (u-ia)^2 z^2
	- i \theta (u-ia) z
	+ \frac{1-\alpha}{\kappa} \,.
\end{equation*}
Let us first observe that, for every $a$ in the NTS analyticity strip, we have $\Re(w) > 0$: indeed,
\begin{equation*}
	\Re(w) =
	\frac{1}{2} \sigma^2 (u^2-a^2) z^2
	- \theta a z
	+ \frac{1-\alpha}{\kappa}
	\geq 
	- \frac{1}{2} \sigma^2 a^2 z^2
	- \theta a z
	+ \frac{1-\alpha}{\kappa}
	> 0 \,,
\end{equation*}
where the last inequality coincides with the analyticity condition in Lemma \ref{prop:nts_strip}.
Then, as $\abs{u} \rightarrow \infty$,
\begin{equation*}
	\Re \left( w^\alpha \right)
	=
	\abs{w}^{\alpha} \cos\left[ \alpha \atan\left( \frac{\Im(w)}{\Re(w)} \right)
	\right]
	=
	\left(
		\frac{1}{2} \sigma^2 z^2
	\right)^\alpha
	\abs{u}^{2\alpha}
	+
	o(\abs{u}^{2\alpha}) \,.
\end{equation*}
By computing the integral in \eqref{eq:ounts_psi}, 
we deduce that the behaviour for large $\abs{u}$ is
\begin{equation*}
	\int_{e^{-bt}}^{1}
		\left(
			\frac{1}{2} \sigma^2 z^2
		\right)^\alpha
		\abs{u}^{2\alpha}
		\frac{ \dd{z} }{z}
	+
	o(\abs{u}^{2\alpha})
	=
	\left(
		\frac{\sigma^2}{2}
	\right)^\alpha
	\frac{1 - e^{-2\alpha b t}}{2\alpha}
	\abs{u}^{2\alpha} 
	+
	o(\abs{u}^{2\alpha}) \,.
\end{equation*}

\smallskip
In the FA case, we study the modulus of the conditional CF $\phi_Z^{\conditional}$ in \eqref{eqn:fa_decomp}, namely
\begin{equation} \label{eq:modulus_of_phiv}
	\abs{\phi_Z^{\conditional}(u-ia,t)} = \frac{ \abs{e^{f(u-ia,t)} - 1} }{e^{\Lambda(t)} - 1} \,,
\end{equation}
where
\begin{equation*}
	f(u, t) := 
    \Lambda(t) \, \phi_J(u,t) = 
	\left( \frac{1-\alpha}{\kappa} \right)^{1-\alpha}
	\frac{1}{\abs{\alpha} b} \,
	\int_{e^{-bt}}^1
		\left(
			\frac{1}{2} \sigma^2 u^2 z^2
			- i \theta u z
			+ \frac{1-\alpha}{\kappa}
		\right)^\alpha
		\frac{\dd{z}}{z} \,.
\end{equation*}
By continuity of the complex function $z \mapsto |e^z - 1|$ at $z = 0$, it follows that, as $\abs{u} \to \infty$,
\begin{equation} \label{eq:modulusexpansion}
	\abs{e^{f(u-ia, t)} - 1} 
	= 
	\abs{f(u-ia, t)} + o(\, \abs{f(u-ia, t)} \,)
	\,.
\end{equation}
%
%The behaviour of modulus of the conditional CF $\phi_Z^{\conditional}$ is then obtained via
%\eqref{eq:modulus_of_phiv} and \eqref{eq:modulusexpansion}.
%
%In particular, for any $a$ in the analyticity strip, the real and imaginary part of $f(u-ia, t)$ have the following asymptotic expansions for large $\abs{u}$:
As in the IA case, it is straightforward to deduce that
--~for every $a$ in the analyticity strip~-- 
the asymptotic behaviours of the real and imaginary parts of $f(u-ia, t)$ for large $\abs{u}$ are:
\begin{equation*}
\begin{cases}
	\displaystyle
	\Re(f(u-ia,t))
	=
	\left( \frac{\sigma^2}{2} \right)^{\alpha}
	\left( \frac{1-\alpha}{\kappa} \right)^{1-\alpha}
	\frac{e^{-2 \alpha b t}-1}{2 \alpha^2 b}
	\abs{u}^{2\alpha}
	+
	o( \abs{u}^{2\alpha} ) \,,
	\\[0.1cm]
	\displaystyle
	\Im(f(u-ia,t))
	=	
	o( \abs{u}^{2\alpha} ) \,.
	\vphantom{	\left( \frac{1-\alpha}{\kappa} \right)^{1-\alpha}}
\end{cases}
\end{equation*}
The thesis follows by substituting these expressions into \eqref{eq:modulus_of_phiv} and \eqref{eq:modulusexpansion}.
$\hfill\square$

\bigskip
\noindent
{\sffamily\bfseries Proof of Proposition \ref{prop:numericalerrordecaytsnts}}

\noindent
Using Algorithm \ref{alg:simulMC1} for IA processes and Algorithm \ref{alg:simul2} for FA processes, we observe distinct asymptotic behaviours of the CF moduli: exponential decay for the former and power-law decay for the latter, as proved in 
Lemmas \ref{prop:ounts_decay}, \ref{prop:outs_decay}, \ref{prop:tsou_decay}, and \ref{prop:ntsou_decay}.
The numerical error $\mathcal{E}(x)$ exhibits a decay rate consistent with that of the CFs, as proved in Proposition \ref{prop:numerical_error1}.
$\hfill\square$

\setcounter{equation}{0}
\setcounter{theorem}{0}

\section{Lévy-driven OU processes of TS and NTS type} \label{app:lcf_of_ts_nts}

{
\color{my_colour}
In this Appendix, we characterise the four main Lévy-driven OU processes
considered in the literature, which are based on two well-known Lévy processes --~TS and NTS~-- and the associated self-decomposable distributions. In the following, we report the LCFs of these processes, derived using \eqref{eq:oulevy-relationship} and \eqref{eq:levyou-relationship}.\footnote{
    For brevity, we omit the discussion of special parameter cases, namely 
    $\alpha_p, \alpha_n \in \{0,1\}$ for OU-TS and TS-OU processes, and 
    $\alpha=0$ for OU-NTS and NTS-OU processes.
}
	%Special parameter cases, namely 
	%$\alpha_p, \alpha_n \in \{0,1\}$ for OU-TS and TS-OU processes, and 
	%$\alpha=0$ for OU-NTS and NTS-OU processes,
    % are deferred to the Supplementary Material.
%
%In Appendices \ref{sec:appendix_TS} (for the TS) and \ref{sec:appendix_NTS} (for the NTS), we report the key quantities used in the simulation 
%--~Lévy triplet, analyticity strip and cumulants~-- as well as some additional properties 
%--~such as self-decomposability and classification.

%The four processes discussed below will serve as concrete benchmarks to illustrate the accuracy % and convergence 
%of our general simulation method, and to compare its computational performance with the existing ED schemes.

%\subsubsection{OU-TS and TS-OU Processes}
%
%\label{ssec:outs_tsou}
%
\bigskip \noindent \,(i)
An \textbf{OU-TS process} is an OU-Lévy whose driver is a TS process.
It is defined by 7 parameters: $\alpha_p, \alpha_n < 2$, $\beta_p, \beta_n > 0$,
$c_p, c_n \geq 0$, and $\gamma_c \in \Erre$.
The LCF of the process $Z_t$ in \eqref{eq:soluz_strong} reads
%is obtained from \eqref{eq:oulevy-relationship}. When $\alpha_p, \alpha_n \in \mathcal{A}_{TS}$, it reads:
%
\begin{equation} \label{eq:outs_chfun}
	\begin{aligned}
		\Psi_Z(u,t) =
    	iu \frac{1 - e^{-bt}}{b} \gamma_c
    	&+
    	\frac{c_p \beta_p^{\alpha_p} \Gamma(-\alpha_p)}{b} \,
    	\left[
        	\int_{\beta_p}^{ \beta_p e^{bt}} 
        	\frac{(z-iu)^{\alpha_p}}{z^{\alpha_p+1}} \dd{z} -
        	bt +
        	\frac{\alpha_p}{\beta_p} iu (1-e^{-bt})
    	\right] 
    	+
    	\\
    	&+
    	\frac{c_n \beta_n^{\alpha_n} \Gamma(-\alpha_n)}{b}
    	\left[
        	\int_{\beta_n}^{\beta_n e^{bt}} 
        	\frac{(z+iu)^{\alpha_n}}{z^{\alpha_n+1}} \dd{z} -
        	bt -
        	\frac{\alpha_n}{\beta_n} iu (1-e^{-bt})
    	\right] \,,
	\end{aligned}
\end{equation}
valid when both $\alpha_p, \alpha_n \notin \{0,1\}$. 
%The cases $\alpha_p, \alpha_n \in \{0,1\}$ are discussed in the Supplementary Material.
\noindent
As noted by \citet{sabino2022cgmy}, 
the integrals in \eqref{eq:outs_chfun} can also be expressed via the Gaussian hypergeometric function ${}_2{F}_1(a,b,c;x)$, for which fast computational routines exist.

\bigskip \noindent \,(ii)
A \textbf{TS-OU process} is a Lévy-OU whose stationary distribution is a TS.
It is defined by 7 parameters: $\alpha_p, \alpha_n \in [0, 2)$, $\beta_p, \beta_n > 0$,
$c_p, c_n \geq 0$, and $\gamma_c \in \Erre$.
The LCF of the process $Z_t$ is % in \eqref{eq:soluz_strong} reads
% is obtained from \eqref{eq:levyou-relationship}. When $\alpha_p, \alpha_n \in \mathcal{A}_{TS}$, it reads:
%
\iffalse
\begin{equation} \label{eq:tsou_chfun}
	\begin{gathered}
		\Psi_Z(u,t) = 
		i u \left( 1 - e^{-bt} \right) 
		\left[ 
			\gamma_{c} 
			- c_p \Gamma(1-\alpha_p) \beta_p^{\alpha_p-1}  
			+ c_n \Gamma(1-\alpha_n) \beta_n^{\alpha_n-1}  
		\right]
		+ 
		\\
		+
		c_p \Gamma(-\alpha_p)
		\left[
			\left( \beta_p - iu \right)^{\alpha_p}
			- \left( \beta_p - iue^{-bt} \right)^{\alpha_p} \,
		\right]
		+ 	
		c_n \Gamma(-\alpha_n)
		\left[
			\left( \beta_n + iu \right)^{\alpha_n}
			- \left( \beta_n + iue^{-bt} \right)^{\alpha_n}
		\right] \,,
	\end{gathered}
\end{equation}
\fi
\begin{equation} \label{eq:tsou_chfun}
	\begin{aligned}
		\Psi_Z(u,t) = 
		i u \left( 1 - e^{-bt} \right) 
		&
		\left[ 
			\gamma_{c} 
			- c_p \Gamma(1-\alpha_p) \beta_p^{\alpha_p-1}  
			+ c_n \Gamma(1-\alpha_n) \beta_n^{\alpha_n-1}  
		\right]
		+ 
		\\
		+ \,
		c_p \Gamma(-\alpha_p)
		&
		\left[
			\left( \beta_p - iu \right)^{\alpha_p}
			- \left( \beta_p - iue^{-bt} \right)^{\alpha_p} \,
		\right]
		+ \,
		\\
		+ 	
		c_n \Gamma(-\alpha_n)
		&
		\left[
			\left( \beta_n + iu \right)^{\alpha_n}
			- \left( \beta_n + iue^{-bt} \right)^{\alpha_n}
		\right] \,,
	\end{aligned}
\end{equation}
valid when both $\alpha_p, \alpha_n \notin \{0,1\}$.
The TS-OU is defined only for non-negative $\alpha_p$ and $\alpha_n$, as the TS distribution is self-decomposable only in this parameter range 
(see, e.g., Lemma \ref{prop:ts_selfdecomp}).

%A detailed classification of both the OU-TS and TS-OU processes into 
%finite-activity, finite-variation and infinite-variation 
%--~depending on the parameters $\alpha_p$ and $\alpha_n$~-- is provided in Appendix \ref{sec:appendix_TS}.

%\subsubsection{OU-NTS and NTS-OU Process}
%\label{ssec:ounts_ntsou}
\bigskip \noindent \,(iii)
An \textbf{OU-NTS process} is a OU-Lévy whose driver is an NTS process.
It is defined by 4 parameters: $\alpha < 1$, $\kappa > 0$, 
$\sigma > 0$, and $\theta \in \Erre$.
The LCF of the process $Z_t$ is %in \eqref{eq:soluz_strong} reads
%
%When $\alpha \in \mathcal{A}_{NTS}$.
%
\begin{equation} \label{eq:ounts_psi}
	\Psi_Z(u,t)
	=
	\frac{1-\alpha}{\kappa \alpha} \, t
	-
	\left( \frac{1-\alpha}{\kappa} \right)^{1-\alpha}
	\frac{1}{\alpha b} \,
	\int_{e^{-bt}}^1
		\left(
			\frac{1}{2} \sigma^2 u^2 z^2
			- i \theta u z
			+ \frac{1-\alpha}{\kappa}
		\right)^\alpha
		\frac{\dd{z}}{z} \,,
\end{equation}
valid when $\alpha \neq 0$. 
%hen $\alpha = 0$, the driver reduces to a Variance Gamma process \citep[VG;][]{madan1990}; we provide further details in the Supplementary Material.
%
Importantly, the case $\theta \neq 0$ is not considered in the literature on ED algorithms.
\citet{sabino2023normal} treats only the case
$\theta = 0$, 
calling it
the \textit{symmetric} case.
%With our methodology, simulating asymmetric OU-NTS (and, similarly, NTS-OU) is straightforward, as discussed in Section \ref{sec:applications}.

\bigskip \noindent \,(iv)
Finally, an \textbf{NTS-OU process} is a Lévy-OU whose stationary distribution is an NTS.
It is defined by 4 parameters: $\alpha \in [0,1)$, $\kappa > 0$, $\sigma > 0$, and $\theta \in \Erre$.
The LCF of the process $Z_t$ is
\begin{equation} \label{eq:ntsou_psi}
	\Psi_Z(u,t)
	=
	\frac{1-\alpha}{\kappa \alpha}
	\left[
		\left(
			1 - 
			\frac{i \kappa}{1-\alpha}
			\left( 
				\theta u e^{-bt} + i \frac{\sigma^2 u^2 e^{-2bt}}{2} 
			\right)
		\right)^{\alpha}
		-
		\left(
			1 - 
			\frac{i \kappa}{1-\alpha}
			\left( 
				\theta u + i \frac{\sigma^2 u^2}{2} 
			\right)
		\right)^{\alpha}
	\right] \,,
\end{equation}
valid when $\alpha \neq 0$.
The NTS-OU is defined only for non-negative $\alpha$, since the NTS distribution is self-decomposable only in this parameter range 
(see, e.g., Lemma \ref{prop:nts_selfdecomp}).
%When $\alpha = 0$, the stationary law reduces to a VG distribution; we provide further details in the Supplementary Material.
Existing simulation methods for this process are also limited to the $\theta = 0$ case \citep[cf.][]{sabino2022exact}.

%A detailed classification of both the OU-NTS and NTS-OU processes into 
%finite-activity, finite-variation, and infinite-variation classes 
%--~depending on the parameter $\alpha$~--
%is provided in Appendix \ref{sec:appendix_NTS}.
}

\clearpage \newpage

\setcounter{equation}{0}
\setcounter{theorem}{0}

\section{The Tempered Stable process} \label{sec:appendix_TS}

The TS process is a widely used Lévy process that encompasses several well-known subcases such as the CGMY process \citep{carr2002}.
It is obtained by exponential tempering of stable Lévy densities, which ensures finite moments of all orders while preserving high flexibility for financial modelling 
\citep[cf.][Section 4.5]{cont2003}.

\begin{definition} \label{def:TSprocess}
A TS process is a Lévy process with the following characteristic triplet 
$(0, \nu, \gamma)$\,\textnormal{:}
	\begin{equation} \label{eq:nu_tsprocess}
	 \begin{cases}
	 \displaystyle
		\nu(x)
		=
		c_p \frac{e^{-\beta_p x}}{x^{1+\alpha_p}}
			\mathbbm{1}_{x > 0}
		+
		c_n \frac{e^{-\beta_n \abs{x}}}{\abs{x}^{1+\alpha_n}}
			\mathbbm{1}_{x < 0} \,,
		\\[0.3cm]
		\displaystyle
		\gamma
		=
		\gamma_c -
   		 c_p \beta_p^{\alpha_p-1} \, \Gamma_U (1-\alpha_p, \beta_p) +
    	 c_n \beta_n^{\alpha_n-1} \, \Gamma_U (1-\alpha_n, \beta_n) \,,
    	 \phantom{\frac{1}{2}} %pareggiare spazio verticale
	\end{cases}	
	\end{equation}
	where $\alpha_p, \alpha_n < 2$, 
	$\beta_p, \beta_n > 0$, 
	$c_p, c_n \geq 0$, 
	$\gamma_c \in \mathbb{R}$ are the model parameters, 
	and $\Gamma_U(\alpha, x)$ denotes the upper incomplete gamma function.
	For $\alpha_p, \alpha_n \notin \{0,1\}$, its characteristic exponent reads
	\begin{equation} \label{eq:char_tsprocess}
	\begin{aligned}
		\psi(u)
		=
		i u \gamma_{c} 
		&+ 
		c_p  \Gamma(-\alpha_p)
		\left[
			\left( \beta_p - iu \right)^{\alpha_p} 
			- \beta_p^{\alpha_p}
			+ i u \alpha_p \, \beta_p^{\alpha_p-1}
		\right]
		+
		\\[0.1cm]
		&+ 	c_n \Gamma(-\alpha_n)
		\left[
			\left( \beta_n + iu \right)^{\alpha_n} 
			- \beta_n^{\alpha_n}
			- i u \alpha_n \, \beta_n^{\alpha_n-1}
		\right] \,.
	\end{aligned}
	\end{equation}
\end{definition}

\medskip
The following two lemmas establish important properties of the TS process that are used throughout the paper. 
The first defines the analyticity strip of the process, which determines the valid range of complex shifts for our simulation method 
(see Section \ref{sec:thealgorithm}). 
\vspace{-0.3\baselineskip}
\begin{lemma} \label{prop:ts_strip}
	For a TS process, the analyticity strip of the CF is given by
	\begin{equation}
		(-p_+, p_-) = (-\beta_p, \beta_n) \,.
	\end{equation}
	In the unilateral positive case (i.e.\ when $c_n = 0$), the strip reduces to $(-\beta_p, +\infty)$.
\end{lemma}
\vspace{-\baselineskip}
\begin{proof}
	See, e.g., \citet{kuchler2013}, Lemma 2.8, p.7
\end{proof}

\smallskip
The second lemma provides exact expressions for the cumulants of the TS process, which are used to test the accuracy of our simulation method in Section~\ref{sec:applications}.
\vspace{-0.3\baselineskip}
\begin{lemma} \label{prop:TScumulants}
	Let $L_t$ be a TS process. Then, the cumulants of $L_1$ are given by:
\begin{equation} \label{eq:TScumulants}
    c_k(L_1) = 
        \begin{dcases}
            \gamma_c
            & \text{if} ~~ k = 1  \,,
            \\
            c_p \beta_p^{\alpha_p - k} \Gamma(k-\alpha_p) + 
            (-1)^k \, 
            c_n \beta_n^{\alpha_n - k} \Gamma(k-\alpha_n)
            & \text{if} ~~ k \geq 2  \,.
        \end{dcases}
\end{equation}
\end{lemma}
\begin{proof}
	See, e.g., \citet{sabino2022exact}, p.21.
\end{proof}

\subsection{Additional properties of OU-TS and TS-OU}

While TS processes describe continuous-time trajectories, the distribution at a fixed time is often the primary object of interest in applications. The TS distribution is defined as the law of a TS process at time $t=1$.
\vspace{-0.3\baselineskip}
\begin{definition}
	A r.v.\ $X$ follows a
	TS
	distribution if its LCF is $\psi(u)$ in \eqref{eq:char_tsprocess}.
	Equivalently, the r.v.\ $X$ has the same law as	$L_1$, where $L_t$ is a TS process.
\end{definition}

The following lemma discusses the self-decomposability of the TS distribution, which is key for the construction of OU-Lévy processes (see Section \ref{sec:theprocesses}).
\vspace{-0.3\baselineskip}
\begin{lemma} \label{prop:ts_selfdecomp}
	The TS distribution is self-decomposable if and only if both
	$\alpha_p$ and $\alpha_n$ are non-negative.
\end{lemma}
\vspace{-\baselineskip}
\begin{proof}
	The result follows directly from Theorem 15.11 in \citet[][p.95]{sato1999}.
\end{proof}

\smallskip
The Lévy-driven OU processes of TS type defined in Appendix \ref{app:lcf_of_ts_nts} exhibit distinct behaviours depending on the stability parameters $\alpha_p$ and $\alpha_n$. Following the literature \citep[see, e.g.,][Section 3.5]{cont2003}, we classify these processes by their activity and variation properties. This classification is crucial for our methodology, as it determines the selection of the simulation algorithm
%and directly impacts numerical error estimation 
(see Section \ref{sec:thealgorithm}).
\vspace{-0.3\baselineskip}
\begin{lemma} \label{prop:ts_classification}
The TS, OU-TS and TS-OU processes can be classified, based on the values of $\alpha_p$ and $\alpha_n$, as in the following table:
\begin{center}
\begin{tabularx}{0.95\textwidth}{|c||Y|Y|Y|} 
	%\toprule
	\hline
	\multirow{2}*{} &
	\multirow{2}*{\bf{Finite-Activity}} &
	\multicolumn{2}{c|} {\bf{Infinite-Activity}}
	\\ \cline{3-4}
	& & \bf{Finite-Variation} & 
	\bf{Infinite-Variation}
	\\
	\hhline{|=|=|=|=|}
	$\vphantom{\Big(}$ \bf{TS} $\vphantom{\Big)}$ & 
			$\small{\max\{ \alpha_p, \alpha_n \} < 0}$ &
			$\small{\max\{ \alpha_p, \alpha_n \} \in [\,0,1)}$ &
			$\small{\max\{ \alpha_p, \alpha_n \} \in [\,1,2)}$
		\\ \hline
	$\vphantom{\Big(}$ \bf{OU-TS} $\vphantom{\Big)}$ & 
			$\small{\max\{ \alpha_p, \alpha_n \} < 0}$ &
			$\small{\max\{ \alpha_p, \alpha_n \} \in [\,0,1)}$ &
			$\small{\max\{ \alpha_p, \alpha_n \} \in [\,1,2)}$
		\\ \hline
		$\vphantom{\Big(}$ \bf{TS-OU} $\vphantom{\Big)}$ &
			$\small{\alpha_p = \alpha_n = 0}$ &
			$\small{\max\{ \alpha_p, \alpha_n \} \in (\,0,1)}$ &
			$\small{\max\{ \alpha_p, \alpha_n \} \in [\,1,2)}$
		\\ \hline
\end{tabularx}
\end{center}
\end{lemma}
\begin{proof}	
	The TS density is given in \eqref{eq:nu_tsprocess}; those of OU-TS and TS-OU are derived from \eqref{eq:oulevymeasure} and \eqref{eq:levyoumeasure}. The classification follows from the standard definitions of activity and variation \citep[see, e.g.,][]{sato1999}.
\end{proof}

\bigskip
Lastly, we establish in the following lemmas the asymptotic behaviour of the CF for the OU-TS and TS-OU processes, which complements the results presented in Section~\ref{ssec:errorcontrol}.
For simplicity, we focus on the unilateral positive case ($c_n=0$); the generalisation to the bilateral case is straightforward.
\vspace{-0.3\baselineskip}
\begin{lemma} \label{prop:outs_decay}
	For an OU-TS process with $\alpha_p \in (0,1) \cup (1,2)$,
	the CF $\phi_Z$ has exponential decay 
    as
	$\abs{u} \to \infty$\,\textnormal{:}
	%Specifically,
	\begin{equation*} \begin{aligned} \label{decay:outs_decay}
		\log \abs{ \phi_Z(u - ia,t) }
		=
		c_p \Gamma(-\alpha_p )
		\cos	
		\left(
			\alpha_p \frac{\pi}{2}
		\right)
		\frac{1-e^{-\alpha_p b t}}{\alpha_p b}
		\abs{u}^{\alpha_p}
		+
		o( \abs{u}^{\alpha_p} )
		\hspace{1.2cm}
		\forall  a \in (-\infty, \beta_p) \,.
	\end{aligned} \end{equation*}
	
\noindent
	For an OU-TS process with $\alpha_p < 0$, 
	the conditional CF $\phi_Z^{\conditional}$ has power-law decay as 
	$\abs{u} \to \infty$\,\textnormal{:}
	%Specifically,
	\begin{equation*}  \label{decay:outs_decay_fa}
		\phantom{log}
		\abs{ \phi_Z^{\conditional}(u - ia,t) }
		=
		\frac{c_p \Gamma(-\alpha_p )}{e^{\Lambda(t)} - 1}
		\,
		\frac{1-e^{-\alpha_p b t}}{\alpha_p b}
		\abs{u}^{\alpha_p}
		+
		o( \abs{u}^{\alpha_p} )
		\hspace{3cm}
		\forall  a \in (-\infty, \beta_p) \,,
	\end{equation*}	
	where $\Lambda(t) = c_p \beta_p^{\alpha_p} \Gamma(-\alpha_p) t$.
\end{lemma}

\begin{lemma} \label{prop:tsou_decay}
	For a TS-OU process with $\alpha_p \in (0,1) \cup (1,2)$, 
	the CF $\phi_Z$ has exponential decay as 
	$\abs{u} \to \infty$\,\textnormal{:}
	\begin{equation*}  \label{decay:tsou_decay}
		\log \abs{\phi_Z(u - ia, t)}
		= 
		c_p \Gamma(-\alpha_p)
		\cos		
		\left(
			\alpha_p \frac{\pi}{2}
		\right)
		\left( 1-e^{-\alpha_p b t} \right)
		\abs{u}^{\alpha_p} 
		+
		o( \abs{u}^{\alpha_p} )
		\hspace{0.8cm}
		\forall  a \in (-\infty, \beta_p) \,.
	\end{equation*}
\end{lemma}
\vspace{-0.8\baselineskip}
\begin{proof}
	The proofs of both lemmas are analogous to Lemma \ref{prop:ounts_decay}. 
\end{proof}

\clearpage \newpage

\setcounter{equation}{0}
\setcounter{theorem}{0}

\section{The Normal Tempered Stable process} \label{sec:appendix_NTS}
 
The NTS process $L_t$ is a time-changed Lévy process of the form
$L_t = \theta S_t + \sigma W_{S_t}$, where 
$\theta \in \Erre$, 
$\sigma>0$,
$W_t$ is a standard BM,
and
$S_t$ is a TS subordinator\footnote{
\setstretch{1.1}
The TS subordinator $S_t$ (i.e.\ a positive unilateral TS process) is parametrised using only two parameters, $\alpha < 1$ and $\kappa > 0$, by setting $\alpha_p = \alpha$, $\beta_p = (1-\alpha)/\kappa$, $c_p = ((1-\alpha)/\kappa)^{1-\alpha}/\Gamma(1-\alpha)$ and $\gamma_c = 1$
%This parametrisation ensures that $\mathbb{E}[S_t] = t$ and $\text{Var}(S_1) = \kappa$ 
\citep[cf.][p.~116]{cont2003}.
}
%In this framework, the TS process $S_t$ acts as a stochastic clock, randomly accelerating or decelerating the Brownian motion $W_t$, and introducing discontinuous moves into an otherwise continuous path
\citep[cf.][Section 4.4]{cont2003}.
Throughout this work, we consider the most general specification of the NTS process, in which the stability parameter $\alpha$ is allowed to take any value in $(-\infty, 1)$, ensuring consistency with the TS process in Appendix \ref{sec:appendix_TS}.

\begin{definition}
%	An NTS process is a Lévy process obtained by
%	time-changing a BM, with drift $\theta \in \Erre$ and
%	volatility $\sigma>0$, by a TS subordinator. 
%	The process has the following characteristic triplet 
%$(0, \nu, \gamma)$\,\textnormal{:}
    An NTS process is a Lévy process with the following characteristic triplet 
$(0, \nu, \gamma)$\,\textnormal{:}
	\begin{equation} \begin{cases} \label{eq:nu_for_nts}
		\displaystyle
		\nu(x)
		=
		\frac{ C(\alpha, \kappa, \sigma, \theta) }
			{\abs{x}^{\alpha+1/2}}
		\,
		e^{\theta x / \sigma^2}
		K_{\alpha + 1/2} 
		\left( \,
			\frac{\abs{x} \, A(\alpha, \kappa, \sigma, \theta)}
           {\sigma^2} \,
		\right) \,,
        %
%		\nu(x)
%		=
%		C(\alpha, \kappa, \sigma, \theta) 
%       \,
%		{\abs{x}^{-\alpha-1/2}}
%		\,
%		e^{\theta x / \sigma^2}
%		K_{\alpha + 1/2} 
%		\left( 
%			\,{\abs{x} \, A(\alpha, \kappa, \sigma, \theta)} 
%            \,/\,
%           {\sigma^2}
%		\right) \,,
		\\[0.50cm]
		\displaystyle
		\gamma
		=
		\frac{1}{\Gamma(1-\alpha)}
			\left( \frac{1-\alpha}{\kappa} \right)^{1-\alpha}
		\int_{0}^{+\infty}
			\frac{ \dd{t} }{ t^{\alpha + 1}  } \,
			e^{-\frac{(1-\alpha)t}{\kappa}}
		\int_{-1}^{1}
			\dd{x}
			\frac{x}{ \sqrt{2\pi t \sigma^2} }
			\, 
			e^{-\frac{ (x-\theta t)^2}{2 t \sigma^2}} \,,
	\end{cases} \end{equation}
	where $\alpha < 1$, $\kappa > 0$, $\theta \in \mathbb{R}$, $\sigma > 0$ are the model parameters,
	 $K_{m}(x)$ is the modified Bessel function of the second kind with order $m$, and
	\begin{align}
		&
		C(\alpha, \kappa, \sigma, \theta) :=
		\frac{2}{\Gamma(1-\alpha) \sigma \sqrt{2\pi}}
		\,
		\left(
			\frac{1-\alpha}{\kappa}
		\right)^{1-\alpha}
		\,
		A(\alpha, \kappa, \sigma, \theta)^{\alpha+1/2} \,\,,
		\label{eq:ccc_coefficient}
		\\[0.15cm]
		&
		A(\alpha, \kappa, \sigma, \theta)
		:=
		\sqrt{  
            \theta^2 + 
			%\frac{2 \sigma^2 (1-\alpha)}{\kappa}}  \,\,.
            {2 \sigma^2 (1-\alpha)}\,/ \,{\kappa} \vphantom{\big(} }
            \,\,.
		\label{eq:aaa_coefficient}
	\end{align}
	For $\alpha \neq 0$, its characteristic exponent reads
	\begin{equation}\label{eq:char_ntsprocess}
		\psi(u)
		=
		\frac{1-\alpha}{\kappa \alpha}
		\left[
			1 -
			\left(
				1 - 
				\frac{i \kappa}{1-\alpha}
				\left( 
					\theta u + i \frac{u^2 \sigma^2}{2} 
				\right)
			\right)^{\alpha} \,
		\right] \,.
	\end{equation}
\end{definition}

\bigskip
The next two lemmas state fundamental properties of the NTS process. The first defines its analyticity strip, establishing the range of valid complex shifts for our simulation method in Section \ref{sec:thealgorithm}.
\vspace{-0.3\baselineskip}
\begin{lemma} \label{prop:nts_strip}
	The analyticity strip of the NTS process is given by
	\begin{equation}
		(-p_+, p_-) 
		=
		\left(
		- \frac{A(\alpha, \kappa, \sigma, \theta) - \theta}{\sigma^2}, \,
		\frac{A(\alpha, \kappa, \sigma, \theta) + \theta}{\sigma^2}
		\right)
		\,.
	\end{equation}
\end{lemma}
\vspace{-\baselineskip}
\begin{proof}
	See Appendix \ref{app:additional_proofs_NTS}.
\end{proof}

\medskip
The second lemma derives a closed-form expression for the NTS cumulants for general order $k$, which, to our knowledge, has not appeared explicitly in the literature before. 
%The cumulants $c_k$ are used to assess the accuracy of our simulation results. % in Section~\ref{sec:applications}.
%
\vspace{-0.3\baselineskip}
\begin{lemma} \label{prop:NTScumulants}
	Let $L_t$ be an NTS process. Then, the cumulants of $L_1$ are given by:
	\begin{equation} \label{eq:ntscumulantsck}
		c_k(L_1) = 
			\sum_{n=0}^{\left \lfloor{k/2}\right \rfloor }
				\frac{k!}{n! \, (k-2n)!} \,\,
			\theta^{k-2n} 
			\left( \frac{\sigma^2}{2} \right)^{n}
			\frac{ \Gamma\left( k - n - \alpha \right) }
					{\Gamma(1-\alpha)}
			\left(
				\frac{\kappa}{1-\alpha} \,
			\right) ^ {k - n - 1} \,,
	\end{equation}
	where the usual convention $0^0 = 1$ applies.
\end{lemma}
\vspace{-\baselineskip}
\begin{proof}
	See Appendix \ref{app:additional_proofs_NTS}.
\end{proof}
\quad

\subsection{Additional properties of OU-NTS and NTS-OU}

Following the standard convention for Lévy processes, the NTS distribution is defined as the law of an NTS process at time $t=1$.
\vspace{-0.3\baselineskip}
\begin{definition}
	A r.v.\ $X$ follows an NTS
	distribution if its CF is
	$\phi_X(u) = e^{\psi(u)}$,
	with $\psi(u)$ in \eqref{eq:char_ntsprocess}.
	Equivalently, the r.v.\ $X$ has the same law as
	$L_1$, where $L_t$ is an NTS process.
\end{definition}

The next lemma addresses the self-decomposability of the NTS distribution, which is essential for the construction of OU–Lévy processes (see Section \ref{sec:theprocesses}).
\vspace{-0.3\baselineskip}
\begin{lemma} \label{prop:nts_selfdecomp}
	The NTS distribution is self-decomposable if and only if
	$\alpha$ is non-negative.
\end{lemma}
\vspace{-\baselineskip}
\begin{proof}
	See, e.g., \citet[][Section 4, p.8]{barndorff2001normal}.
\end{proof}

\smallskip
The Lévy-driven OU processes of NTS type described in Appendix \ref{app:lcf_of_ts_nts} exhibit distinct behaviours depending on the stability parameter $\alpha$.
The following lemma classifies these processes by activity and variation, directly informing algorithm selection and error control (see Section~\ref{sec:thealgorithm}).

\vspace{-0.3\baselineskip}
\begin{lemma} \label{prop:nts_classification}
The NTS, OU-NTS and NTS-OU processes can be classified based on the values of $\alpha$ as in the following table:

\begin{center}
\begin{tabularx}{0.95\textwidth}{|c||Y|Y|Y|} 
	%\toprule
	\hline
	\multirow{2}*{} &
	\multirow{2}*{\bf{Finite-Activity}} &
	\multicolumn{2}{c|} {\bf{Infinite-Activity}}
	\\ \cline{3-4}
	& & \bf{Finite-Variation} & 
	\bf{Infinite-Variation}
	\\
	\hhline{|=|=|=|=|}
	$\vphantom{\Big(}$ \bf{NTS} $\vphantom{\Big)}$ & 
		$\small{\alpha < 0}$ &
		$\small{\alpha \in \left[\,0, \,\tfrac{1}{2} \, \right)}$ &
		$\small{\alpha \in \left[\,\tfrac{1}{2}, 1 \, \right)}$
	\\ \hline
	$\vphantom{\Big(}$ \bf{OU-NTS} $\vphantom{\Big)}$ &
		$\small{\alpha < 0}$ &
		$\small{\alpha \in \left[\,0, \,\tfrac{1}{2} \, \right)}$ &
		$\small{\alpha \in \left[\,\tfrac{1}{2}, 1 \, \right)}$
	\\ \hline
	$\vphantom{\Big(}$ \bf{NTS-OU} $\vphantom{\Big)}$ & 
		$\small{\alpha = 0}$ &
		$\small{\alpha \in \left(\,0, \,\tfrac{1}{2} \, \right)}$ &
		$\small{\alpha \in \left[\,\tfrac{1}{2}, 1 \, \right)}$
	\\ \hline
\end{tabularx}
\end{center}
\end{lemma}
\begin{proof}	
	The NTS density is given in \eqref{eq:nu_for_nts}; those of OU-NTS and NTS-OU are derived from \eqref{eq:oulevymeasure} and \eqref{eq:levyoumeasure}. The classification follows from the standard definitions of activity and variation \citep[see, e.g.,][]{sato1999}.
\end{proof}

\medskip
Lastly, the following lemma characterises the asymptotic behaviour of the CF of the NTS-OU process, complementing the OU-NTS results presented in Section~\ref{ssec:errorcontrol}.
\vspace{-0.3\baselineskip}
\begin{lemma} \label{prop:ntsou_decay}
	For an NTS-OU process with $\alpha \in (0,1)$, the CF $\phi_Z$
	has exponential decay as $\abs{u} \to \infty$\,\textnormal{:}
	\begin{equation*} \label{decay:ntsou_decay}
		\log \abs{ \phi_Z(u - ia,t) }
		=
		-
		\frac{(\sigma^2 / 2)^\alpha }{\alpha}
		\left( \frac{1-\alpha}{\kappa} \right) ^{1-\alpha}
		\left( 1 - e^{-2\alpha b t} \right)
		\abs{u}^{2\alpha}
		+ o \left( \abs{u}^{2\alpha} \right)
		\hspace{1cm}
		\forall  a \in (-p_-, p_+) \,\,
	\end{equation*}
	where $(-p_+, p_-)$ denotes the NTS analyticity strip
    (see Lemma \ref{prop:nts_strip}).
\end{lemma}
\vspace{-\baselineskip}
\begin{proof}
	The proof is analogous to Lemma \ref{prop:ounts_decay}. 
\end{proof}

\clearpage \newpage

\subsection{Proofs for NTS Properties} \label{app:additional_proofs_NTS}

\smallskip
\noindent
{\sffamily\bfseries Proof of Lemma \ref{prop:nts_strip}}

\noindent
By Theorem 3.1 in \citet{lukacs1972}, the analyticity strip is determined by the singularities of the CF on the imaginary axis. Setting $u = i\xi$ in \eqref{eq:char_ntsprocess}, with $\xi \in \mathbb{R}$, the characteristic exponent is well-defined if
\begin{equation*}
	1 + \frac{\kappa}{1-\alpha}\Big(\theta \xi - \tfrac12 \sigma^2 \xi^2\Big) > 0 \,.
\end{equation*}
The boundaries of the strip are obtained by solving this quadratic inequality in $\xi$.
%Solving this quadratic inequality in $\xi$ gives the two extrema, $\xi_1 < 0$ and $\xi_2 > 0$, of the analyticity strip.
%
%\begin{equation*}
%	\frac{\theta-A(\alpha, \kappa, \sigma, \theta)}{\sigma^2}
%	=: 
%	\xi_1
%	< \xi < 
%	\xi_2 
%	:=
%	\frac{\theta+A(\alpha, \kappa, \sigma, \theta)}{\sigma^2}
%\end{equation*}
%where $\xi_1 < 0$ and $\xi_2 > 0$, and $A(\alpha,\kappa,\sigma,\theta)$ is defined in \eqref{eq:aaa_coefficient}.
%
%The Variance Gamma case ($\alpha = 0$, cf.\ Supplementary Material), is analogous.
$\hfill\square$

\bigskip
\noindent
{\sffamily\bfseries Proof of Lemma \ref{prop:NTScumulants}}

\noindent
Let $L_t$ be an NTS process, and $L_1$ an NTS random variable. By definition, the cumulants are given by
\begin{equation} \label{eq:cumulants_faa}
	c_k(L_1)
	= (-i)^k \left. \frac{d^k \psi_{L}(u)}{du^k} \right|_{u=0}
	= (-i)^k \left. \frac{d^k}{du^k} 
		\psi_{S}\!\left(\theta u + i \frac{\sigma^2}{2}u^2\right)
	\right|_{u=0} \,\,,
\end{equation}
where $\psi_{S}(\bigcdot)$ denotes the characteristic exponent of the TS subordinator $S_t$
\citep[cf., e.g.,][Theorem 4.2, p.~115]{cont2003}.

\smallskip
For $k = 1$, direct differentiation of \eqref{eq:char_ntsprocess} gives
$c_1(L_1) = \theta$.

For $k \ge 2$, we set $g(u) := \theta u + i\sigma^2 u^2 / 2$ and apply the Faà di Bruno formula:
\begin{equation*}
	\left. \frac{d^k}{du^k} \psi_{S}(g(u)) \right|_{u=0}
	= \sum_{n=1}^{k} \psi_S^{(n)}(g(0)) \,
	B_{k,n}\big(g'(0), g''(0), \dots, g^{(k-n+1)}(0)\big) \,,
\end{equation*}
where $B_{k,n}$ are the partial Bell polynomials 
\citep[see, e.g.,][]{comtet2012advanced}.
Since $g(u)$ is quadratic, all derivatives of order higher than two are equal to zero.
Hence, for every $k \ge 2$ and $n = 1, \dots, k$, the Bell polynomials reduce to
\begin{equation*}
	B_{k,n}(x_1, x_2, 0, 0, \dots)
	=
	\frac{k!}{(2n-k)!\,(k-n)!\,2^{\,k-n}}\,
	x_1^{\,2n-k}\,x_2^{\,k-n}\,
	\mathbbm{1}_{\{\lceil k/2 \rceil \le n \le k\}} \,.
\end{equation*}
Substituting these expressions into \eqref{eq:cumulants_faa} yields
\begin{equation} \label{eq:cumulants_faa_2}
	c_k(L_1)
	=
	(-i)^k
	\sum_{n=\lceil k/2 \rceil}^{k}
	i^{\,n} c_n(S_1)\,
	\frac{k!}{(2n-k)!\,(k-n)!\,2^{\,k-n}}\,
	\theta^{\,2n-k}\,(i\sigma^2)^{\,k-n},
\end{equation}
where $c_n(S_1) = (-i)^n \psi_S^{(n)}(0)$ are the cumulants of the TS subordinator,
given in \eqref{eq:TScumulants}.
Under the TS subordinator parametrisation, these cumulants read
\begin{equation*}
	c_n(S_1)
	= \left(\frac{\kappa}{1-\alpha} \right)^{n-1} \,
	\frac{\Gamma(n-\alpha)}{\Gamma(1-\alpha)} \,.
\end{equation*}
Plugging this expression into \eqref{eq:cumulants_faa_2}, simplifying the powers of $i$, and
reindexing the summation via $n \mapsto k-n$, we obtain the stated closed-form result.
$\hfill\square$

\clearpage
\newpage
\setstretch{1.0}

\section*{Symbols and Acronyms}

\begin{center}

\begin{tabular} {|c|l|}
	\toprule
	\textbf{Symbol}& \textbf{Description}\\ \bottomrule
	$L_t$ & Background driving Lévy process \\
	$X_t$ & Lévy-driven OU process defined in \eqref{eq:soluz_strong}\\
	$Z_t$ & Integral Process defined in \eqref{eq:soluz_strong}\\
	$Z_t^{\conditional}$ & For FA processes, \textit{conditional} law of $Z_t$, 
		cf.\ \eqref{eq:conditional_decomposition_of_Zt} \\
	$W_t$ & Brownian Motion \\
	$\gamma$ & Drift in the characteristic triplet \\
	$\nu_L(x)$ & Lévy density of a Lévy process $L_t$ \\
	$\phi_L(x)$ & Characteristic Function of a Lévy process $L_t$ \\
	$\psi_L(x)$ & Characteristic Exponent of a Lévy process $L_t$ \\
	$\nu_Z(x, t)$ & Lévy density of an additive process $Z_t$ \\
	$\phi_Z(x,t)$ & Characteristic Function of an additive process $Z_t$ \\
	$\Psi_Z(x,t)$ & Log-Characteristic Function of an additive process $Z_t$ \\
	$p_-, p_+$ & Boundaries of the analyticity strip \\
	$c_k(Y)$ & Cumulant of order $k$ of a r.v.\ $Y$ \\
	$P(x)$  & True CDF of the innovation \\
	$\widehat{P}(x)$  & Numerical of the CDF of the process innovation \\
	$\mathcal{E}(x)$ & Numerical error in the CDF of the innovation \\
	$a$ & Imaginary shift of the integration path in the complex plane \\
	$h$ & Grid step in the Fourier domain \\
	$M$ & Integer number equal to $\log_2(N)$, with $N$ number of FFT points \\
	$N$ & Number of points in the FFT grid, equal to $2^M$ \\
	$\alpha$ & Stability parameter for TS and NTS process/distribution \\
	$b$ & Mean-reversion parameter in the OU equation \\
	$\beta, c, \gamma_c$ & Other parameters of TS process/distribution \\
	$\kappa, \theta, \sigma$ & Other parameters of NTS process/distribution \\
	$B, \ell, \omega$ & Coefficients for the CF bounds, cf. \eqref{eq:expexpexp} and \eqref{eq:powpowpow} \\
    $\Lambda(t)$ & Intensity of $Z_t$ in the compound Poisson representation \\
	$\Phi(x)$ & Standard Normal CDF \\
	$\Gamma(x)$ & Gamma function \\
	$\Gamma_U(x)$ & Incomplete upper gamma function \\
	$E_1(x)$ & Exponential integral function \\
	$K_m(x)$ & Modified Bessel function of second kind with order $m$ \\
	$_2F_1(a,b;c;x)$ & Gaussian hypergeometric function \\
	$\rchi$ & Moneyness level of an option \\
	$\mathcal{U}(a,b)$ & Uniform distribution on $[a,b]$ \\
	$O(\bigcdot)$ & Landau's \textit{Big-O} symbol \\
	$\,o\,(\bigcdot)$ & Landau's \textit{Little-o} symbol \\
	\bottomrule		
\end{tabular}
	
\end{center}	

\bigskip 

\begin{center}
	\begin{tabular}{|c|l|}
		\toprule
		\textbf{Acronym}& \textbf{Description}\\ \bottomrule
		ATM & At-the-money \\
		BM & Brownian Motion \\
		bp & Basis Point \\
		CDF & Cumulative Distribution Function \\
		CF & Characteristic Function \\
		ED & Exact Decomposition \\
		FA & Finite-Activity \\
		FFT & Fast Fourier Transform \\
		FGMC & Fast and General Monte Carlo \\
		IA & Infinite-Activity \\
		i.i.d. & Independent and Identically Distributed \\
		ITM & In-the-money \\
		LCF & Log-Characteristic Function \\
		MAPE & Mean Absolute Percentage Error \\
		MC & Monte Carlo \\
		%NIG & Normal Inverse Gaussian \\
		NTS & Normal Tempered Stable \\
		OTM & Out-of-the-money \\
		OU & Ornstein-Uhlenbeck\\
		PDF & Probability Density Function \\
		RMSE & Root Mean Squared Error \\
		r.v. & Random Variable \\
		SD & Standard Deviation \\
		SDE & Stochastic Differential Equation \\
		TS & Tempered Stable \\
		VG & Variance Gamma \\
		\bottomrule
	\end{tabular}
\end{center}

\end{document}